\newif\ifcomments % enable comments
\newif\ifauthors  % enable printing authors
\newif\iftrapdoor % enable the extended paper with trapdoor one-way functions; deprecated

\commentsfalse
\authorstrue
\trapdoortrue

\documentclass[11pt,pdfa]{article}%
\pdfoutput=1 %See https://tex.stackexchange.com/questions/182523/publishing-to-arxiv-with-hyperlinks for more information

\usepackage{iftex}
\ifPDFTeX
  \usepackage[utf8]{inputenc}
  \usepackage[noTeX]{mmap}
  \usepackage[T1]{fontenc}
\fi
\ifLuaTeX
  \usepackage{luatex85}
  \usepackage[noTeX]{mmap}
\fi

\usepackage{amsmath,amsfonts,amssymb}
\usepackage{amsthm}
\usepackage{mathtools}
\usepackage{graphicx}
\usepackage[margin=1in]{geometry} % makes smaller bottom margins than fullpage package, but otherwise the same
\usepackage[pagebackref,linktocpage]{hyperref}
\hypersetup{
    colorlinks=true, % false: boxed links; true: colored links
    linkcolor=blue, % color of internal links
    citecolor=blue, % color of links to bibliography
    urlcolor=blue % color of external links
}
\usepackage{microtype}
\usepackage{thm-restate}
\usepackage[capitalize,nameinlink]{cleveref}

\usepackage{tikz}
\usepackage{pgfplots}
\pgfplotsset{compat=1.13}
\usepackage{wrapfig}
\usepackage{enumerate}
\usepackage{braket}
\usepackage{qcircuit}
\usepackage{circledsteps} %Text in circles
\usepackage{dsfont} %Bold numbers

\usepackage{thm-restate}

\renewcommand{\backref}[1]{}

\renewcommand{\backrefalt}[4]{%
\ifcase #1 %
\or
[p.\ #2]%
\else
[pp.\ #2]%
\fi}

\newtheorem{theorem}{Theorem}

\newtheorem{claim}[theorem]{Claim}

\newtheorem{corollary}[theorem]{Corollary}

\newtheorem{definition}[theorem]{Definition}

\newtheorem{fact}[theorem]{Fact}
\newtheorem{lemma}[theorem]{Lemma}

\newtheorem{problem}[theorem]{Problem}
\newtheorem{proposition}[theorem]{Proposition}

\Crefformat{enumi}{#2Item~(#1)#3}
\crefformat{enumi}{Item~(#2#1#3)}

\ifcomments
  \newcommand{\comment}[1]{{\color{red}#1}}
\else
  \newcommand{\comment}[1]{}
\fi

\usepackage{mleftright}
\newcommand{\mparen}[1]{\mleft(#1\mright)}
\newcommand{\mbracket}[1]{\mleft[#1\mright]}
\newcommand{\mbrace}[1]{\mleft\{#1\mright\}}
\newcommand{\abs}[1]{\left|#1\right|}

\usepackage{tabularx}
\newcolumntype{Y}{>{\centering\arraybackslash}X|}
\definecolor{mylightgreen}{RGB}{217,234,211}
\definecolor{mylightred}{RGB}{244,204,204}
\usepackage{colortbl}
\usepackage{hhline}
\newcommand{\Ucell}{\cellcolor{mylightred} Uniform}
\newcommand{\Fcell}{\cellcolor{mylightgreen} Forrelated}

\newcommand{\equad}{\mathrel{\phantom{=}}}
\newcommand{\E}{\mathop{\mathbf{E}}}
\newcommand{\Naturals}{\mathbb{N}}

\newcommand{\poly}{\mathrm{poly}}
\newcommand{\negl}{\mathrm{negl}}
\newcommand{\polylog}{\mathrm{polylog}}
\newcommand{\quasipoly}{\mathrm{quasipoly}}
\newcommand{\AND}{\normalfont\textsc{AND}} 
\newcommand{\OR}{\normalfont\textsc{OR}}
\newcommand{\NOT}{\normalfont\textsc{NOT}}

\newcommand{\Dist}{\mathcal{D}}
\newcommand{\PDist}{\mathcal{P}}
\newcommand{\SDist}{\mathcal{S}}

\newcommand{\Forr}{\mathcal{F}}

\DeclareMathOperator{\s}{\mathsf{s}}

\begin{document}

\title{Quantum-Computable One-Way Functions \\ without One-Way Functions}

\newcommand{\email}[1]{\href{mailto:#1}{\texttt{#1}}}

\ifauthors
\author{William Kretschmer\thanks{ Simons Institute for the Theory of Computing, University of California, Berkeley. \ Email:
\email{kretsch@berkeley.edu}. \ Supported by the U.S.\ Department of Energy, Office of Science, National
Quantum Information Science Research Centers, Quantum Systems Accelerator. } \and 
Luowen Qian\thanks{NTT Research, Inc. \ Email: \email{luowen@qcry.pt}. \ This work was done in part when LQ was at Boston University and supported by DARPA under Agreement No.\ HR00112020023.}
\and
Avishay Tal\thanks{University of California, Berkeley. \ Email:
\email{atal@berkeley.edu}. Supported by a Sloan Research Fellowship and NSF CAREER Award CCF-2145474.}
}
\fi

\date{}
\maketitle

\begin{abstract}
    We construct a classical oracle relative to which $\mathsf{P} = \mathsf{NP}$ but quantum-computable quantum-secure trapdoor one-way functions exist.
    This is a substantial strengthening of the result of Kretschmer, Qian, Sinha, and Tal (STOC 2023), which only achieved single-copy pseudorandom quantum states relative to an oracle that collapses $\mathsf{NP}$ to $\mathsf{P}$. 
    For example, our result implies multi-copy pseudorandom states and pseudorandom unitaries, but also classical-communication public-key encryption, signatures, and oblivious transfer schemes relative to an oracle on which $\mathsf{P}=\mathsf{NP}$.
    Hence, in our new relativized world, classical computers live in ``Algorithmica'' whereas quantum computers live in ``Cryptomania,'' using the language of Impagliazzo's worlds.
    %\comment{Q: maybe add signatures as well? Previous oracle separations technically also have key exchange (QKD) and OT but no standard signatures afaic}
    %\comment{A: Should we add: In particular, our result implies multi-copy pseudorandom quantum states relative to an oracle on which $\mathsf{P}=\mathsf{NP}$.}

    Our proof relies on a new distributional block-insensitivity lemma for $\mathsf{AC^0}$ circuits, wherein a single block is resampled from an arbitrary distribution.
    \comment{maybe there's a better way to phrase this...}
\comment{Potential title: Quantum-Computable Classical Cryptomania in Algorithmica}
\comment{A: How about `Quantum Cryptomania in Algorithmica'}

\end{abstract}

\section{Introduction}

%\comment{Q: merged changes from introduction-with-trapdoor.tex}
What cryptography would survive in the event of someone or some organization discovering a practical algorithm to solve $\mathsf{NP}$-complete problems?
Unfortunately, almost all computationally-secure classical cryptography relies on the existence of one-way functions \cite{IL89-owf,Gol90-prg}.
Thus, their security certainly requires at least $\mathsf{P} \neq \mathsf{NP}$.%, as otherwise guessing a secret key becomes trivial.
%\comment{A: Can't we say it more directly: Unfortunately, almost all computationally-secure classical cryptography requires at least $\mathsf{P} \neq \mathsf{NP}$, as otherwise guessing a secret key becomes trivial.}

Recent works have %established oracle separations which
hinted that quantum analogues of many important cryptographic tasks may not be subject to this barrier.
A series of black-box separations established that \textit{pseudorandom quantum states}---a quantum counterpart to classical pseudorandom generators---can exist relative to oracles that make $\mathsf{NP}$~\cite{KQST23-prs}, $\mathsf{QMA}$~\cite{Kre21-pseudorandom}, and more powerful complexity classes~\cite{LMW24-synthesis} computationally easy.
Combined with parallel efforts to build cryptosystems from pseudorandom states~\cite{AQY22-prs,MY22-prs,AGQY22-prfs,ALY24-pseudorandom,BBOSS24-botprf,CGG24-qccc}, we now know that useful computationally-secure quantum cryptography could conceivably exist in a world where $\mathsf{P} = \mathsf{NP}$.

Upon closer inspection, the cryptographic protocols realized in these oracle separations all require quantum communication, or even long-term quantum memory, in addition to fault-tolerant quantum computation.
%Particularly, the oracle separation by Krestchmer~\cite{Kre21-pseudorandom} shows that quantum cryptography might exist even if $\mathsf{BQP} = \mathsf{QMA}$.
So, these oracle separations are hardly a satisfactory replacement for the classical cryptography that we currently use.
For example, it is possible that we find efficient algorithms for problems like breaking SHA-3 and learning with errors \textit{before} we have quantum internet or reliable quantum storage.

Even outside of this nightmare scenario, quantum communication is not always desirable.
Besides the obvious challenges of engineering a robust quantum channel, there are other theoretical and practical limitations to protocols that use quantum communication.
Certain scenarios require broadcasting, which is impossible quantumly due to the no-broadcasting theorem~\cite{BCF+96-broadcast}.
For a concrete example, public-key encryptions and digital signatures with public key infrastructure (PKI) have been essential for securing digital communications.
While quantum versions of public-key encryptions from one-way functions \cite{BGH+23-qpke,Col23-qtd,KMNY24-tamper,MW24-robust} and digital signatures without one-way functions \cite{GC01-qds,MY22-prs} have been explored, these schemes all make use of an \textit{uncloneable} quantum public key, besides only satisfying weak security requirements such as one-time security.
This unclonability is undesirable from the PKI perspective because it prevents the PKI from distributing the public key: once the copies that the PKI holds are exhausted, new users would be unable to obtain public keys and take part in the protocols.%\footnote{While (efficiently) cloneable and classical communication could be technically different~\cite{NZ24-telegraph}, in this work we focus on classical communication for simplicity.} % removing this because too much information

%By inverting the secret quantum states from the classical transcript, $\mathsf{BQP} = \mathsf{QMA}$ would rule out various cryptographic protocols even if they use quantum computers but a classical communication channel.
%However, the other oracle separation of Kretschmer, Qian, Sinha, and Tal~\cite{KQST23-prs} only separates them from $\mathsf P = \mathsf{NP}$, which does not immediately rule out most classical-communication quantum cryptography.
Altogether, these limitations of quantum communication in cryptography give rise to the following natural question:
\begin{center}
    \emph{What quantum cryptography with \emph{classical communication} is still possible if $\mathsf P = \mathsf{NP}$?}
\end{center}

\subsection{Main Result}

Consider a \emph{quantum-computable} one-way function (OWF), which is like an ordinary one-way function except that instead of mandating an efficient classical evaluation algorithm, we permit (pseudo-deterministic) quantum algorithms as well.
Security must hold against both classical and quantum adversaries.
It is clear that this is a weakening of traditional (quantum-secure) one-way functions.
Since the only difference is in permitting evaluation by a quantum computer, one might be optimistic that this object is not so different from one-way functions, perhaps by employing some clever dequantizations.
%\comment{Q: adding this sentence}
After all, a $\mathsf{BPP}$-computable pseudo-deterministic one-way function can be derandomized through standard black-box techniques.\footnote{In particular, a $\mathsf{BPP}$-computable OWF $f(x; r)$ gives a distributional OWF $f'(x, r) := f(x; r)$, which can then be used to construct a standard OWF \cite{IL89-owf}.}

In this work, we show, surprisingly, that quantum-computable OWFs can exist in an oracle world where $\mathsf{P} = \mathsf{NP}$, and therefore dequantizing a quantum-computable OWF is impossible in a black-box fashion.
In fact, our main theorem is a significant strengthening of this, where we construct quantum-computable \textit{trapdoor} one-way functions that are consistent with $\mathsf{P} = \mathsf{NP}$.

\begin{theorem}[\Cref{thm:main_towf}, informal]
    \label{thm:main_informal}
    There exists a classical oracle relative to which $\mathsf P = \mathsf{NP}$ and a quantum-computable trapdoor one-way function exists.
    Furthermore, the trapdoor one-way function has pseudorandom public keys.
\end{theorem}

By invoking known post-quantum fully black-box reductions, we obtain that relative to the same classical oracle as \Cref{thm:main_informal}, the following \emph{classical-communication} cryptographic schemes also exist:
\begin{itemize}
    \item Public-key encryptions with semantic security. (\Cref{cor:cryptomania})
    \item Public-key signatures with existential unforgeable security. (\cite[Section 5.1]{Son14-pqc})%\comment{why is this a cite and not a corollary?}
    \item Oblivious transfer protocols with simulation security. (\Cref{cor:cryptomania})
\end{itemize}
Therefore, it appears that in this oracular world, classical computers live in ``Algorithmica'' while quantum computers live in ``Cryptomania'' \cite{Imp95-average}, even without the need of any long-term quantum memory or quantum communication!

\comment{How do we feel about adding this paragraph?}
As a corollary, all of these cryptographic schemes are separated from $\mathsf P \neq \mathsf{NP}$ (thus OWFs) as well.
Note that prior to our work, even mildly quantum variants of these were not known to be separated from OWFs (e.g., public-key encryption with quantum ciphertext, or quantum signatures with standard security).

\paragraph{Implications for Quantum Pseudorandomness.}
Recall that the previous oracle separations of comparable nature are instantiations of \textit{pseudorandom quantum states} relative to oracles that make classical cryptography easy~\cite{Kre21-pseudorandom,KQST23-prs,LMW24-synthesis}.
Our result is strictly stronger, then, because quantum-computable one-way functions are also sufficient to construct pseudorandom states since these reductions to one-way functions are fully-black-box~\cite{JLS18-prs,BS19-binary,AGQY22-prfs}.
%\comment{A: "these reductions to one-way functions are fully-black-box" seems like we are referring to something mentioned before. Where was it mentioned? If not, how about changing the phrasing to "construct pseudorandom states using fully-black-box reductions". W: I like this change}
Specifically, one can use the now standard binary phase construction:
\[
\ket{\psi_k} \coloneqq \frac{1}{\sqrt{2^n}}\sum_{x \in \{0,1\}^n} (-1)^{f_k(x)}\ket{x},
\]
where $\{f_k\}$ is a keyed family of pseudorandom functions constructed from the one-way functions~\cite{Zha21-qprf}.
This directly answers an open problem from~\cite{KQST23-prs}:
\begin{corollary}
    There exists a classical oracle relative to which $\mathsf P = \mathsf{NP}$ and many-copy-secure pseudorandom states exist.
\end{corollary}
For comparison, \cite{KQST23-prs} only achieved \textit{single-copy-secure} pseudorandom states.
Following more recent advances, one can build even more versatile quantum pseudorandomness primitives from one-way functions, including pseudorandom states of arbitrary polynomial length~\cite{BS20-scalable} and even pseudorandom unitaries~\cite{MPSY24-pru,CBB+24-pru,MH24-pru}.
As a corollary, these also exist relative to our oracle as well.

\subsection{Technical Overview}

\paragraph{Warmup: Separating Quantum-Computable PRFs.}
We first show how to construct an oracle relative to which $\mathsf{P} = \mathsf{NP}$ and quantum-computable \textit{pseudorandom functions} (PRFs) exist.
\comment{W: Do we need to define PRFs? Q: I think it's fine without}
The analysis in this special case is conceptually simpler, though still rich enough to capture most of the important ideas needed to generalize to trapdoor one-way functions.

The main idea behind this oracle separation is to construct a random oracle in a special encoding that is only accessible by $\mathsf{BQP}$ but not $\mathsf{PH}$.
This encoding technique was previously seen in the work of Aaronson, Ingram, and Kretschmer~\cite{AIK21-acrobatics}, although similar ideas have also appeared in even earlier works~\cite{BM99-relativized,ABK16-cheat}.
%\comment{A: Is this the first instantiation of this idea? I thought something of this form is already in ABK15 -- Separations in query complexity using cheat sheets}
Then intuitively, the quantum-computable pseudorandom functions can simply be constructed by direct evaluation of this quantum-computable random oracle.

We now explain the oracle construction in more detail.
Similar to the oracle used in~\cite{KQST23-prs}, our oracle $\mathcal{O}$ can be thought of as a pair of oracles $(A, B)$.
The oracle $A$ encodes the $\mathsf{BQP}$-accessible random oracle, and the addition of auxiliary oracle $B$ has the effect of making $\mathsf{P} = \mathsf{NP}$.
Morally speaking, $B$ behaves as if it were an oracle for $\mathsf{PH}^A$.
Thus, showing the security of the quantum-computable PRFs relative to $\mathcal{O}$ amounts to proving security against polynomial-time quantum adversaries that can query any $\mathsf{PH}^A$ language.
For brevity, we'll call these $\mathsf{BQP^{PH}}$ (oracular) adversaries in this exposition.

The key difference between our oracle construction and that of~\cite{KQST23-prs} is in how we build $A$.
In~\cite{KQST23-prs}, $A$ is simply a random oracle, whereas in this warmup, $A$ is an \textit{encoding} of a uniformly random oracle.
As with many of the oracles constructed in~\cite{AIK21-acrobatics}, we encode each bit of the oracle using the \textit{Forrelation} problem~\cite{Aar10-bqp-ph,RT19-bqp-ph}.
Recall that Forrelation (in its broadest sense) is the following task: given query access to a pair of functions $f, g : \{0,1\}^\ell \to \{0,1\}$, distinguish between
\begin{itemize}
    \item[(NO)] $f$ and $g$ are independent uniformly random functions, or
    \item[(YES)] $f$ and $g$ are individually random, but sampled in such a way that $f$ is noticeably correlated with the Boolean Fourier transform of $g$ (i.e., $f$ and $g$ are ``Forrelated'').
\end{itemize}
In contrast to~\cite{KQST23-prs}, whose analysis required a carefully-crafted version of Forrelation, we only require black-box use of one key fact from \cite{RT19-bqp-ph}: that these two distributions are efficiently distinguishable by $\mathsf{BQP}$ algorithms, but not by $\mathsf{PH}$ algorithms.

Letting $L$ be the random oracle, our strategy for encoding $L$ is to hide each bit of its output behind an instance of Forrelation.
That is, for each $x \in \{0,1\}^*$, we add a region of $A$ that encodes a pair of uniformly random Boolean functions if $L(x) = 0$, or otherwise a pair of Forrelated functions if $L(x) = 1$.
It is straightforward to show that oracle access to $A$ enables a quantum algorithm to recover any bit of $L$.
The candidate quantum-computable PRF family $\{f_k\}$ is the natural choice defined by $f_k(x) \coloneqq L(k||x)$. \comment{A: this notation is confusing, as we just defined $L(x)$. W: what about $L(k||x)$?}

The main technical difficulty is to formalize the intuition that this encoding is secure enough against $\mathsf{PH}$ that distinguishing the PRFs from random remains hard even against $\mathsf{BQP^{PH}}$ adversaries.
We expect this to be the case because to $\mathsf{PH}$, the oracle $A$ looks completely random and independent of $L$, and therefore a $\mathsf{BQP^{PH}}$ adversary should not have much more power than a $\mathsf{BQP}$ adversary.
To make this argument rigorous, we have to prove that a $\mathsf{BQP^{PH}}$ adversary, given oracle access to an auxiliary function $h$, cannot distinguish whether $h$ is uniformly random or whether $h$ is one of the pseudorandom functions $f_k$.
We establish this by observing that it suffices to show the following: if the adversary is given a uniformly random $h$, then the adversary is unlikely to detect a change to the oracle $A$ so as to make it consistent with $h = f_k$ for a random key $k$.

For the PRF family $\{f_k\}_{k \in \{0,1\}^n}$ of functions $f_k: \{0,1\}^n \to \{0,1\}$, we can view the part of $A$ encoding these functions as a $2^n \times 2^n$ matrix of Forrelation instances.
In this matrix, the rows are indexed by keys $k \in \{0,1\}^n$, the columns by inputs $x \in \{0,1\}^n$, and the corresponding Forrelation instance is Forrelated if $f_k(x) = 1$, or uniform otherwise.
Given $h: \{0,1\}^n \to \{0,1\}$, the goal of the adversary is to determine whether there is a row of the matrix whose pattern of Forrelated/uniform instances is consistent with $h$.
See \Cref{fig:prf_distinguishing_task} for an example.

% \begin{figure}
%     \centering
%     \scriptsize
%     \begin{tabularx}{0.12\textwidth}{|*{4}{Y}}
%     \hline
%     \Ucell\\\hline
%     \Ucell\\\hline
%     \Ucell\\\hline
%     \Ucell\\\hline
%     \end{tabularx}
%     \quad vs.\quad
%     \begin{tabularx}{0.12\textwidth}{|*{4}{Y}}
%     \hline
%     \Ucell\\\hline
%     \Ucell\\\hline
%     {\bf\Fcell}\\\hline
%     \Ucell\\\hline
%     \end{tabularx}
%     \caption{In the $\textsc{OR} \cdot \textsc{Forrelation}$ problem, either all of the input blocks are uniform (left), or whether a single block is Forrelated (right).}
%     \label{fig:or_forrelation}
% \end{figure}

\begin{figure}
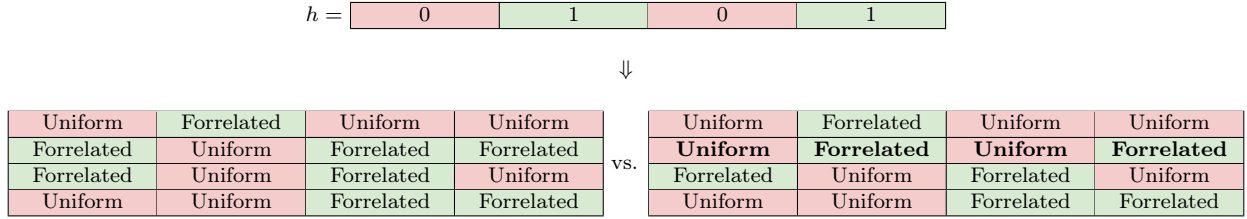

    \centering
    \scriptsize
    $h =$ \begin{tabularx}{0.48\textwidth}{|*{6}{Y}}
    \hline
    \cellcolor{mylightred} 0 &
    \cellcolor{mylightgreen} 1 &
    \cellcolor{mylightred} 0 &
    \cellcolor{mylightgreen} 1\\\hline
    \end{tabularx}
    \[
    \Downarrow
    \]
    \begin{tabularx}{0.48\textwidth}{|*{4}{Y}}
    \hline
    \Ucell & \Fcell & \Ucell & \Ucell\\\hline
    \Fcell & \Ucell & \Fcell & \Fcell\\\hline
    \Fcell & \Ucell & \Fcell & \Ucell\\\hline
    \Ucell & \Ucell & \Fcell & \Fcell\\\hline
    \end{tabularx}
    vs.
    \begin{tabularx}{0.48\textwidth}{|*{4}{Y}}
    \hline
    \Ucell & \Fcell & \Ucell & \Ucell\\\hline
    \bf\Ucell & \bf\Fcell & \bf\Ucell & \bf\Fcell\\\hline
    \Fcell & \Ucell & \Fcell & \Ucell\\\hline
    \Ucell & \Ucell & \Fcell & \Fcell\\\hline
    \end{tabularx}
    \caption{An example of the distinguishing task for our quantum-computable PRFs with $n = 2$.
    The goal is to decide (left) whether the Forrelation instances in the matrix are all randomly assigned to either uniform or Forrelated, or (right) whether the pattern of 0s/1s in $h$ matches the pattern of uniform/Forrelated in one of the rows.}
    \label{fig:prf_distinguishing_task}
\end{figure}

Our security proof proceeds in a similar fashion to the $\mathsf{BQP^{PH}}$ lower bound for the so-called $\textsc{OR} \circ \textsc{Forrelation}$ problem that was considered in prior works~\cite{AIK21-acrobatics,KQST23-prs}.
In $\textsc{OR} \circ \textsc{Forrelation}$, we are given a list of instances of the Forrelation problem, and must decide whether they are all uniform (NO), or whether a single one of the instances is Forrelated (YES).
Using the well-known correspondence between $\mathsf{PH}$ query algorithms and $\mathsf{AC^0}$ circuits~\cite{FSS84-circuit-oracle}, the $\mathsf{BQP^{PH}}$ lower bound for $\textsc{OR} \circ \textsc{Forrelation}$ reduces to a certain type of sensitivity concentration result for $\mathsf{AC^0}$ circuits.
The key step~\cite[Lemma 45]{AIK21-acrobatics} shows that for most uniformly random $M \times N$ Boolean matrices, an $\mathsf{AC^0}$ circuit is unlikely to notice the change if we uniformly swap out one of the rows for fresh random bits.

We show that security of our PRF ensemble also reduces to a comparable statement about the sensitivity of $\mathsf{AC^0}$ circuits, but under a distribution of Boolean matrices that is not uniform.
Instead, we have to consider matrices like those sampled in \Cref{fig:prf_distinguishing_task}, where the blocks are a random pattern of uniform or Forrelated.
For this purpose, we are able to show the following main technical lemma, which informally states the following.
Consider a quasi-polynomial-size $\mathsf{AC^0}$ circuit that takes $K$ blocks of bits, where each block is sampled independently from an arbitrary distribution $\Dist$.
Then for sufficiently large $K$, it is hard for it to notice if any block is resampled from the same distribution $\Dist$.
More formally,

\begin{lemma}[\Cref{lem:ac0_dist_block_sensitivity}, restated]
\label{lem:ac0_dist_block_sensitivity_informal}
Let $f: \{0,1\}^{KM} \to \{0,1\}$ be an $\mathsf{AC^0}$ circuit of size $s$ and depth $d$. Let $\Dist$ be a distribution over $\{0,1\}^M$. %Let $x \in \{0,1\}^{KM}$ be an input, viewed as a $K$-by-$M$ matrix. %with $K$ rows and $M$ columns. 
Let $x\sim \Dist^K$ be an input to $f$, viewed as a $K \times M$ matrix.
%\comment{A:I changed this sentence to emphasize that $x$ is a random variable. Otherwise it seems as if $x$ was a fixed string.}
Let $y$ be sampled depending on $x$ as follows: uniformly select one of the rows of $x$, resample that row from $\Dist$, and leave the other rows of $x$ unchanged. Then for any $p > 0$:
\[
\Pr_{x \sim \Dist^K}\left[\Pr_y\left[f(x) \neq f(y)\right] \ge p \right] \le \frac{4K}{p} \cdot 2^{-\frac{p K}{O(\log s)^{d - 1}}}.
\]
\end{lemma}

Notably, this holds for \textit{any} distribution $\Dist$, and not merely the distribution of randomly chosen uniform/Forrelated blocks.
Previously, \Cref{lem:ac0_dist_block_sensitivity_informal} was only known to hold in the special case where $\Dist$ is the uniform distribution~\cite[Lemma 45]{AIK21-acrobatics}.

%\comment{Q: Since this is the step in the proof where the cleverness goes into, we should probably say more. Feel free to revise further}
The proof of this lemma involves a careful construction of a related $\mathsf{AC^0}$ circuit from $f$, whose sensitivity corresponds to the probability of $f$ noticing the block being resampled.
This then allows us to relate this to known sensitivity bounds on $\mathsf{AC^0}$ circuits~\cite{AIK21-acrobatics}.
We refer interested readers to \Cref{lem:ac0_dist_block_sensitivity} for the details.

\paragraph{Generalization to Trapdoor Functions.}
Since \Cref{lem:ac0_dist_block_sensitivity_informal} morally lets us resample an arbitrary block that can be arbitrarily distributed, we can in fact prove that the distinguishing task in \Cref{fig:prf_distinguishing_task} is hard for any distribution of functions instead of just uniform.
A natural idea then is to use Forrelation to encode a more structured oracle that allows us to construct more structured cryptographic primitives.
However, we also cannot introduce more structure than what \Cref{lem:ac0_dist_block_sensitivity_informal} allows us to handle.
For example, to get quantum-computable oblivious transfer protocols, a natural idea would be to encode random trapdoored permutations or something similar instead.
However, unlike a random function, a random permutation is already problematic since each entry in a random permutation is weakly correlated with the other entries.
On the other hand, \Cref{lem:ac0_dist_block_sensitivity_informal} only works with product distributions.

Inspired by this example, we instead start with random functions and only introduce just enough structure to have a trapdoor.
Specifically, to obtain an oracle relative to which some form of quantum-computable trapdoor one-way functions exist but $\mathsf{P} = \mathsf{NP}$, we utilize a similar Forrelation encoding, but the oracle $A$ is no longer random.
Instead, $A$ encodes a triple of functions $G, F, I$ where
\begin{itemize}
    \item $G(td)$ is a random function mapping a trapdoor $td$ to its public key $pk$;%\comment{A: why 'still'?}
    \item $F(pk, x)$ is also a random function mapping a public key $pk$ and input $x$ to an output $y$;
    \item Finally, $I(td, y)$ is the only ``structured'' function that inverts an image $y$ of $F(G(td), \cdot)$ using the trapdoor $td$.
\end{itemize}
Then the quantum-computable trapdoor one-way function construction is simply evaluating these three functions $G, F, I$ by a $\mathsf{BQP}$ algorithm that queries the Forrelation encodings in $A$.

For security, we want to show that a $\mathsf{BQP^{PH}}$ adversary, given $pk = G(td)$ and $y = F(pk, x)$ for random trapdoor $td$ and input $x$, cannot find a preimage $x'$ such that $F(pk, x') = y$.
Proving this is more challenging than before since the encoded oracle here is more sophisticated than simply a random function.
However, notice that the previous security proof for the PRF reduces to a resampling indistinguishability task (\Cref{fig:prf_distinguishing_task}).
%\comment{(W: removing this) similar to \cite{BBBV97-search} for $\mathsf{BQP}$.
%Therefore, one idea is to first prove that $(G, F, I)$ gives a $\mathsf{BQP}$-secure trapdoor one-way function using only BBBV-resampling kinds of argument, then we could hope to lift this argument into security against $\mathsf{BQP^{PH}}$ following the template of the PRF against $\mathsf{BQP^{PH}}$ security proof.
%Informally, such a $\mathsf{BQP}$ security proof can be shown as follows:}
We can similarly reduce security to resampling indistinguishability with some additional steps:

\begin{enumerate}
    \item Starting with the real experiment $(G, F, I)$, we first argue that $G$ and $I$ do not help the adversary in inverting $y$ as follows. \begin{enumerate}
        \item First, resample $G(td) = pk^*$ and set $I(td, \cdot)$ to be the inversion table for $F(pk^*, \cdot)$ instead: this is indistinguishable since we are resampling one of exponentially many blocks of $(G, I)$.
        \item Next, we want to claim that $I$ no longer contains the inversion table for $F(pk, \cdot)$.
            The idea is to argue that $td$ was the only trapdoor that inverts $F(pk, \cdot)$ and after resampling it, there is no trapdoor left to invert $F(pk, \cdot)$ with overwhelming probability.
            For this to hold, it suffices to take $G$ to be length expanding enough so that it is injective with overwhelming probability.
    \end{enumerate}
    \item After the first step, $F(pk, \cdot)$ is essentially a random function independent of the rest of the oracle, which includes $I$, $G$, and the rest of $F$.
        Therefore, we can prove it is one-way using the same approach as before.
        Morally, we next resample $F(pk, x)$ indistinguishably so that $y$ is no longer in the image of $F(pk, \cdot)$, making inversion impossible.
\end{enumerate}

\paragraph{Beyond Trapdoor Functions.}
Observe that this proof sketch actually establishes something stronger.
Specifically, the two steps above prove that our trapdoor function construction satisfies two additional properties, which are respectively: (a) the public keys are pseudorandom and (b) the function is one-way under a truly random public key.
These properties allow us to construct what is called a fakeable public-key encryption (PKE) scheme \cite{GKMRV00-pkeot} which is ``essentially equivalent'' to semi-honest oblivious transfer.
On a high level, a fakeable public-key primitive has a ``fake mode'' of sampling the public key such that security holds even if given the randomness for fakely sampling the public key.
In our case, the fake mode sampling would simply be outputting the input randomness as is.

% We break the security proof into two steps: first, we show that the public key is pseudorandom, meaning that the adversary cannot distinguish $pk = G(td)$ from a uniformly random public key $pk^*$.
% Then, we argue that $y = F(pk^*, x)$ is hard to invert for a uniformly random public key $pk^*$, by arguing that $y$ is industinguishable from a uniformly random $y^*$.
% In both steps, like the case of pseudorandom functions above, we show that the adversary's task amounts to searching a matrix of Forrelation instances for a row that matches a certain pattern known to the adversary.
% So, at its core, the security proof is the same argument that we used for pseudorandom functions, just repeated twice.

Finally, we briefly comment on adapting constructions for trapdoor one-way functions with pseudorandom public keys to the quantum-computable setting, such as constructing fully-secure oblivious transfer from a semi-honest protocol.
Building cryptography out of quantum-computable primitives requires additional care,
because it is not always possible to mindlessly substitute a classical primitive with a quantum-computable counterpart.
%\comment{Rewrote the below to zero-knowledge instead. More understandable? W: good}
For example, consider the scenario where we wish to prove that a one-way function was computed correctly in zero-knowledge.\footnote{This is not a contrived example: many oblivious transfer protocol constructions do make use of such functionality.}
Classically, this could be done just by the assumption that one-way functions exist because the statement above is in $\mathsf{NP}$.
However, if we wish to instead prove that a \emph{quantum-computable} one-way function was computed correctly in zero-knowledge, then this would appear to be a $\mathsf{QCMA}$ statement, so the construction breaks down.
%it is unknown how to classically\footnote{On the other hand, it is known how to quantumly garble a quantum circuit \cite{BY22-garbled}.} garble a classical circuit that uses a quantum-computable one-way function from the same assumption: it appears that you would need to be able to garble pseudo-deterministic quantum circuits as well!
We resolve this by observing that as long as we have a \emph{post-quantum fully black-box} reduction \cite{RTV04-reduction} then substituting with a quantum-computable primitive works.

\subsection{Discussion}

\paragraph{Explaining Quantum-Classical Separations.}
There are two high-level reasons why quantum cryptographic primitives are harder to break than classical ones, even if they inherently can only be computationally secure.\footnote{We focus on separations of these strictly-computational primitives, unlike statistical-computational separations such as QKD~\cite{BB84-qkd} vs.\ classical key exchange.}
The first is purely complexity-theoretic: because the challenger is quantum rather than classical, adversaries require stronger computational power to detect patterns produced by the challenger.
For example, whereas inverting a one-way function is a canonical example of an $\mathsf{NP}$ problem, inverting a \textit{quantum-computable} one-way function is a $\mathsf{QCMA}$ problem. The containment $\mathsf{NP} \subseteq \mathsf{QCMA}$ is believed to be strict, and this distinction alone allows for the possibility that quantum cryptography could be beyond the grasp of $\mathsf{NP}$ algorithms.

The second reason is information-theoretic: the challenges themselves could be complex and highly-entangled quantum states, rather than classical bit strings that can be readily copied.
So, we cannot even express the security games as problems within our usual mathematical framework of complexity classes like $\mathsf{P}$, $\mathsf{NP}$, $\mathsf{BQP}$, or $\mathsf{QMA}$: these classes are only equipped to operate on classical inputs!
While some recent efforts have been made to define complexity classes that accept quantum inputs (e.g.\ $\mathsf{unitaryBQP}$~\cite{BEMPQY23-uhlmann}), the connections between these ``$\mathsf{unitary}$'' complexity classes and their classical-input counterparts remain unclear.

%\comment{Another example of this is the statistical security of QKD whereas classical key exchange can only be computationally secure; also maybe MNY24 + Qia24 where quantum auxiliary input commitments can exist relative to any oracle but classical auxiliary input ones cannot. Do we want to add these? W: in my opinion no, because those examples do not involve computational complexity}

We point out this distinction because to date, the existing oracle separations between pseudorandom states and classical cryptography~\cite{Kre21-pseudorandom,KQST23-prs,LMW24-synthesis} have widely been understood as arising from the second feature.
Our \Cref{thm:main_informal} is the first that clearly makes use of the first feature exclusively, because quantum-computable one-way functions output classical security challenges.
%\comment{Q: changing ``quantum state'' to ``quantum challenge'' to be more specific about where the state appears}
For comparison, recall:
\begin{itemize}
    \item \cite{Kre21-pseudorandom} constructs a quantum oracle relative to which $\mathsf{BQP} = \mathsf{QMA}$ and pseudorandom states exist.
    The proof makes crucial use of the quantum-challenge nature of pseudorandom states.
    One way to see this is the fact that quantum-computable one-way functions do not exist relative to this oracle.
    \item \cite{KQST23-prs} constructs an oracle relative to which $\mathsf{P} = \mathsf{NP}$ and single-copy pseudorandom states exist.
    Here, it is not so clear whether the quantumness of the pseudorandom states is essential for this separation, as quantum-computable one-way functions may or may not exist relative to this oracle.
    (We discuss this further in the related work section.)
    \item \cite{LMW24-synthesis} (implicitly\footnote{%
        Take a random language $L$, and let the oracle $\mathcal{O}$ be any $\mathsf{PSPACE}^L$-complete language.
        Then clearly $\mathsf{P}^\mathcal{O} = \mathsf{PSPACE}^\mathcal{O}$.
       \cite{LMW24-synthesis} construct a single-copy pseudorandom state ensemble using queries to $L$ (which can certainly be simulated using queries to $\mathcal{O}$) whose parallel-query security holds relative to \textit{any} oracle $\mathcal{O}$.
    }) constructs a classical oracle relative to which $\mathsf{P} = \mathsf{PSPACE}$ and single-copy pseudorandom states cannot be broken by efficient parallel-query adversaries.
    But $\mathsf{P} = \mathsf{PSPACE}$ implies (for example) that any one-round classical-communication falsifiable protocol can be broken by an efficient parallel-query adversary, because the optimal adversary strategy can be simulated in $\mathsf{PSPACE}$.
    This certainly implies that quantum-computable one-way functions do not exist relative to this oracle.
    So, the security of the pseudorandom state ensemble here also relies on its use of quantum challenges.
\end{itemize}

\iffalse
In fact, this even extends to separations between classical and quantum cryptography (below pseudorandom states).
The most famous example of quantum key distribution is statistically secure whereas classical key distribution necessarily implies $\mathsf{P} \neq \mathsf{NP}$.
Similarly, quantum commitments with quantum auxiliary inputs can be unconditionally (computationally) secure yet classical ones with auxiliary inputs also implies $\mathsf{P} \neq \mathsf{NP}$.
Since both of these quantum cryptographic constructions are secure relative to any oracle, they also implicitly give oracle separations by for example, taking a $\mathsf{PSPACE}$-complete language.
This also indirectly shows that these two objects crucially relies on quantum challenges.
\comment{keep?}
\fi

Thus, our work is the only one to use the first feature alone, hinting at the possibility of quantum advantage for computing a one-way function.
For this reason, a worthwhile direction for future research complementing ours is to give stronger evidence that the second feature (the use of quantum states in computationally-secure cryptography) directly enables separations from classical cryptography.
All of the existing separations in this regard have caveats: \cite{Kre21-pseudorandom} uses a quantum oracle, \cite{KQST23-prs} might not rely on the use of quantum challenges at all, and \cite{LMW24-synthesis} only obtains security against parallel-query adversaries.
For this purpose, we reiterate the following open problem that was raised in earlier works~\cite{Kre21-pseudorandom,KQST23-prs}:

\begin{problem}
\label{prob:better_oracle}
    Construct a classical oracle relative to which $\mathsf{P} = \mathsf{QMA}$ (or at least $\mathsf{BQP} = \mathsf{QCMA}$) and pseudorandom states (or at least quantum commitments) exist.
\end{problem}

The main appeal of \Cref{prob:better_oracle} is that it would answer this conceptual question about the role of quantum challenges in cryptography, \textit{without} necessarily requiring a resolution to the long-standing unitary synthesis problem in quantum query complexity~\cite{AK07-qcma-qma,Aar16-barbados,LMW24-synthesis}.

\paragraph{Related Work.}
Compared to~\cite{KQST23-prs}, who gave a black-box construction of single-copy-secure pseudorandom states relative to an oracle that makes $\mathsf{P} = \mathsf{NP}$, our work has some advantages.
As noted before, our result is strictly stronger, because one can additionally build quantum-computable (trapdoor) one-way functions relative to our oracle. 
These imply single-copy pseudorandom states but are not necessary for them~\cite{Kre21-pseudorandom}.
Also, our proof is somewhat simpler, because we do not require a special version of the Forrelation problem; hardness against $\mathsf{AC^0}$ and easiness for quantum algorithms are sufficient for the construction to work.
One difference, however, is that our oracle is more structured.
Both our oracle and~\cite{KQST23-prs}'s take the form $\mathcal{O} = (A, B)$, where in both cases $B$ is constructed from $A$ in an identical fashion. But, whereas~\cite{KQST23-prs} takes $A$ to be a random oracle, we do not.
The \textit{Aaronson--Ambainis conjecture}~\cite{AA14-structure} provides some evidence that this difference is necessary: if the Aaronson--Ambainis conjecture is true, then any pseudo-deterministic quantum algorithm querying the random oracle can be query-efficiently simulated by a classical algorithm as well.
Thus, intuitively, a random oracle does not assist in building a quantum-computable one-way function that is not also classically-computable.

Our main theorem also improves upon the separation in~\cite[Theorem 4]{AIK21-acrobatics}, which shows that there is an oracle relative to which $\mathsf{P} = \mathsf{NP}$ but $\mathsf{BQP} \neq \mathsf{QCMA}$.
The same complexity class separations hold relative to our oracle, because quantum-computable one-way functions cannot exist if $\mathsf{BQP} = \mathsf{QCMA}$.

Starting with the work of Ananth, Gulati, Qian, and Yuen~\cite{AGQY22-prfs}, there have been a few works~\cite{ALY24-pseudorandom,BBOSS24-botprf,CGG24-qccc} aiming to construct classical-communication quantum cryptography without using a one-way function.
One might intuitively believe that these constructions do not rely on one-way functions because they are based on \textit{logarithmic-output-length} pseudorandom states, which seem weaker.
%The main insight is that if the output length is only logarithmic, then tomography becomes efficient, which maps quantum states down to their classical representation in a way that preserves their cryptographic usefulness.
Nevertheless, our work is the first to black-box separate these logarithmic-output-length pseudorandom states from one-way functions, as they can be black-box constructed from quantum-computable pseudorandom functions~\cite{BS20-scalable}.
Furthermore, we separate a stronger object, a (pseudo-deterministic) trapdoor one-way function.
Pseudorandom state-based constructions are also messier and less elegant due to the use of tomography.

Another recent line of works aim to investigate the (im)possibility of constructing quantum public-key encryption (PKE) schemes from one-way functions \cite{ACCFLM22-qro2ka,BGH+23-qpke,Col23-qtd,BGVV24-owfqpke,KMNY24-tamper,LLLL24-owfqpke,MW24-robust}.
Specifically, it is shown that they can be constructed from OWFs if the public key is allowed to be an uncloneable quantum state.
%On the construction front, w
We do consider the stronger notion of quantum-computable classical-communication PKE, but the construction essentially assumes quantum-computable trapdoor OWFs.
This assumption (as we have shown) is incomparable to OWFs.
%From a separation perspective, one may interpret these works as suggesting that quantum-public-key encryption schemes is separated from post-quantum PKE.
%Our work is somewhat stronger in this perspective since we show quantum-computable classical PKE do not even imply classical OWFs.

% While one might believe that these short pseudorandom states might also be separated from one-way functions, known separations can only separate pseudorandom states with at least superlogarithmic output length as noted before.
% In fact, this was not known until our work separates even pseudo-deterministic quantum-computable one-way functions from one-way functions, which is an even stronger separation.
% Combined with the works above, this hints that logarithmic-length pseudorandom states are closer to a $\mathsf{QCMA}$-style primitive, which is separated from both more quantum-native primitives such as super-logarithmic-length pseudorandom states, as well as $\mathsf{NP}$-style primitives such as classical pseudorandomness.

\paragraph{Future Directions.}
One open question is to strengthen our separation to separate quantum-computable collision resistance, quantum-computable one-way permutations, or quantum-computable indistinguishability obfuscation and quantum-computable OWFs from $\mathsf{P} = \mathsf{NP}$ as well.
Our oracle construction could be straightforwardly adapted to those settings by replacing the Forrelation-encoded random oracle with a Forrelation encoding of any oracle that instantiates these primitives.
However, proving security of these primitives relative to such oracles remains a challenge. It seems one would need a stronger $\mathsf{AC^0}$ sensitivity bound (like \Cref{lem:ac0_dist_block_sensitivity_informal}) capable of handling a more complicated resampling procedure, e.g.\ where multiple rows may be updated in a correlated fashion.

While our work, along with earlier oracle separations, only studies separations between abstract cryptographic primitives, these also reveal a potential pathway towards finding cryptographically-useful concrete assumptions beneath one-way functions.
Specifically, this work shows that if we consider a candidate one-way function that we only know how to evaluate with quantum computers, then its security could resist even the proof of $\mathsf P = \mathsf{NP}$.
Unlike the work of~\cite{KQST23-prs}, however, it is unclear how to heuristically instantiate our oracle, because we do not know of any candidates for non-oracular problems in $\mathsf{BQP}$ that are hard for $\mathsf{PH}$~\cite{Aar18-bqp-ph}.
%sampling Forrelated functions is not known to be efficient.\footnote{Indeed, it is even open to give a candidate non-oracular problem in $\mathsf{BQP}$ that is hard for $\mathsf{PH}$~\cite{RT19-bqp-ph}.}
However, the recent work of Khurana and Tomer~\cite{KT24-advantage} constructs quantum cryptography from $\mathsf{\#P}$-hardness and quantum advantage conjectures, and our work hints that quantum-computable classical cryptography could be realized from similar assumptions as well.
%\comment{do we want to cite BHHP24, CGGH24, or HM24 as well?}
We leave as a future research direction to investigate concrete quantum-computable classical cryptography instantiations.

Finally, we note that our technical contributions may prove useful towards resolving a certain open problem about oracle separations of complexity classes.
Aaronson~\cite{Aar10-bqp-ph} raised the question of whether there exists an oracle relative to which $\mathsf{NP} \subseteq \mathsf{BQP}$ but $\mathsf{PH} \not\subset \mathsf{BQP}$.
Later work by Aaronson, Ingram, and Kretschmer~\cite{AIK21-acrobatics} conjectured the possibility of a more granular separation, namely: an oracle relative to which $\mathsf{\Sigma}_k^\mathsf{P} \subseteq \mathsf{BQP}$ but $\mathsf{\Sigma}_{k+1}^\mathsf{P} \not\subset \mathsf{BQP}$, for any desired $k \in \Naturals$.
Moreover, they gave a candidate construction of such an oracle, and sketched a possible route towards showing that this oracle satisfies the desired properties~\cite[Section 6.2]{AIK21-acrobatics}.
We remark that \Cref{lem:ac0_dist_block_sensitivity_informal} is precisely one of the steps in their proposed sketch, although it does not seem to be sufficient on its own to achieve the separation.

\section{Preliminaries}

\subsection{Notation}
For a distribution $\Dist$, $x \sim \Dist$ denotes that $x$ is sampled from $\Dist$.
If $K \in \Naturals$, then $x \sim \Dist^K$ means that $x$ is sampled from the product distribution $\underbrace{\Dist \times \cdots \times \Dist}_{k \text{ times}}$.
When $S$ is a finite set, $x \sim S$ denotes sampling from the uniform distribution over $S$.
An algorithm $\mathcal{A}$ with access to an oracle $\mathcal{O}$ is denoted by %putting the oracle into superscript: 
$\mathcal{A}^{\mathcal{O}}$.

$\mathsf{AC^0}[s,d]$ is the set of Boolean circuits of size at most $s$ and depth at most $d$ that are composed of unbounded-fan-in $\AND$, $\OR$, and $\NOT$ gates (where $\NOT$ gates do not contribute towards depth).

We use $\poly(n)$, $\polylog(n)$, $\quasipoly(n)$, and $\negl(n)$ in their standard fashion.
Formally, $\poly(n)$ indicates an arbitrary polynomially-bounded function of $n$, i.e.\ a function $f$ for which there is a constant $c > 0$ such that $f(n) \le n^c$ for all sufficiently large $n$.
$\polylog(n)$ is an arbitrary $f$ satisfying $f(n) \le \log(n)^c$ for all sufficiently large $n$, and $\quasipoly(n)$ is an arbitrary $f$ satisfying $f(n) \le 2^{\log(n)^c}$ for all sufficiently large $n$. Finally, $\negl(n)$ is an arbitrary negligibly-bounded function of $n$, i.e.\ a function $f$ such that for every $c > 0$, for all sufficiently large $n$, $f(n) \le n^{-c}$.

\subsection{Forrelation}
%\comment{Special notation for Forrelation? Q: do you mean something like $\mathsf{Forrelation}$?}

In this section, we introduce the Forrelation distribution and some variants of it, along with their needed properties.

\begin{theorem}[{\cite[Theorem 1.2 and Claim 8.2]{RT19-bqp-ph}}]
\label{thm:raz-tal}
For all sufficiently large $L$, there exists an explicit distribution $\mathcal{F}_L$ that we call the \emph{Forrelation distribution} over $\{0,1\}^L$ such that:
\begin{enumerate}
\item There exists a quantum algorithm $\mathcal{A}$ that makes $\polylog(L)$ queries and runs in time $\polylog(L)$ such that:
\[
\Pr_{x \sim \mathcal{F}_L} [\mathcal{A}^x = 1] - \Pr_{x \sim \{0,1\}^{L}} [\mathcal{A}^x = 1]  \ge 1 - \frac{1}{L^2}.
\]
\item For any $s = \quasipoly(L)$ and $d = O(1)$, there exists some $\delta' = \frac{\polylog(L)}{\sqrt{L}}$ such that for any $C \in \mathsf{AC^0}[s, d]$:
\[
\left|\Pr_{x \sim \mathcal{F}_L} [C(x) = 1] - \Pr_{x \sim \{0,1\}^{L}} [C(x) = 1]  \right| \le \delta'.
\]
\end{enumerate}
\end{theorem}

Let $\SDist$ be a distribution over $\{0, 1\}^N$.
We define a distribution $\PDist_{\SDist,L}$  (the ``patterned Forrelation distribution'') as the following compound distribution over $\{0,1\}^{NL}$.
First, draw a sample $z \sim \SDist$.
Then, sample row $i$ from $\Forr_L$ if $z_i = 1$, and otherwise sample from $\{0,1\}^L$ if $z_i = 0$.

When $z \in \{0,1\}^N$, we use $\PDist_{z, L}$ as shorthand for $\PDist_{\SDist,L}$ where $\SDist$ is the constant distribution that always samples $z$.
When $N \in \Naturals$, we also use $\PDist_{N, L}$ as shorthand for $\PDist_{\{0,1\}^N,L}$.

%Define $\Dist_{\SDist,L}$ to be the compound distribution over $\{0,1\}^{NL}$ obtained by first drawing $z \sim \SDist$ and then outputting a sample from $\PDist_{z,L}$.
%We also define $\Dist_{N,L}$ as a shorthand for $\Dist_{\UniformDist_N, L}$ where $\UniformDist_N$ is the uniform distribution over $\{0, 1\}^N$.

%Given $z \in \{0,1\}^N$, let $\PDist_{z,L}$ (the ``patterned Forrelation distribution'') be a distribution over $\{0,1\}^{NL}$ where row $i$ is drawn from $\Forr_L$ if $z_i = 1$, or drawn from $\{0,1\}^L$ if $z_i = 0$.

\begin{lemma}
\label{lem:planted_forrelation}
For any $s = \quasipoly(L)$ and $d = O(1)$, there exists some $\delta = \frac{\polylog(L)N}{\sqrt{L}}$ such that for any $C \in \mathsf{AC^0}[s, d]$ and $z \in \{0,1\}^N$:
\[
\abs{\Pr_{x \sim \PDist_{z,L}} [C(x) = 1] - \Pr_{x \sim \{0,1\}^{NL}} [C(x) = 1] } \le \delta.
\]
\end{lemma}

\begin{proof}
    This follows from a standard hybrid argument. For $i = 0, 1, \ldots, N$, let $z^{(i)} \in \{0,1\}^N$ be the string such that the first $i$ bits are the same as $z$, and the last $N - i$ bits are all zeros. Notice that $\PDist_{z^{(0)},L}$ is the uniform distribution over $\{0,1\}^{NL}$, and $\PDist_{z^{(N)},L} = \PDist_{z,L}$. Hence,
    \begin{align*}
        \abs{\Pr_{x \sim \PDist_{z,L}} [C(x) = 1] - \Pr_{x \sim \{0,1\}^{NL}} [C(x) = 1] }
        &= \abs{\Pr_{x \sim \PDist_{z^{(N)},L}} [C(x) = 1] - \Pr_{x \sim \PDist_{z^{(0)},L}} [C(x) = 1] }\\
        &\le \sum_{i=1}^N \abs{\Pr_{x \sim \PDist_{z^{(i)},L}} [C(x) = 1] - \Pr_{x \sim \PDist_{z^{(i-1)},L}} [C(x) = 1] }\\
        &\le \delta' N,
    \end{align*}
    by the triangle inequality and \Cref{thm:raz-tal}.
    We complete the proof by taking $\delta := \delta' N$.
\end{proof}

% Let $\SDist$ be a distribution over $\{0, 1\}^N$.
% Define $\Dist_{\SDist,L}$ to be the compound distribution over $\{0,1\}^{NL}$ obtained by first drawing $z \sim \SDist$ and then outputting a sample from $\PDist_{z,L}$.
% We also define $\Dist_{N,L}$ as a shorthand for $\Dist_{\UniformDist_N, L}$ where $\UniformDist_N$ is the uniform distribution over $\{0, 1\}^N$.

\begin{corollary}
\label{cor:avg_planted_forrelation}
For any $s = \quasipoly(L)$, and $d = O(1)$, there exists some $\delta = \frac{\polylog(L)N}{\sqrt{L}}$ such that for any $C \in \mathsf{AC^0}[s, d]$ and any distribution $\SDist$ over $\{0,1\}^N$:
\[
\abs{\Pr_{x \sim \PDist_{\SDist,L}} [C(x) = 1] - \Pr_{x \sim \{0,1\}^{NL}} [C(x) = 1] } \le \delta.
\]
\end{corollary}

\begin{proof}
    Follows from \Cref{lem:planted_forrelation} and the fact that $\PDist_{\SDist,L}$ is a mixture of $\PDist_{z, L}$ over $z \sim \SDist$.
    %\comment{A: Should this be $z\sim \SDist$? W: yes}
\end{proof}

\begin{corollary}
\label{cor:avg_planted_forrelation_combined}
For any $s = \quasipoly(L)$ and $d = O(1)$, there exists some $\delta = \frac{\polylog(L)N}{\sqrt{L}}$ such that for any $C \in \mathsf{AC^0}[s, d]$ and any two distributions $\SDist,\SDist'$ over $\{0,1\}^N$:
\[
\abs{\Pr_{x \sim \PDist_{\SDist,L}} [C(x) = 1] - \Pr_{x \sim \PDist_{\SDist',L}} [C(x) = 1] } \le \delta.
\]
\end{corollary}

\begin{proof}
    Follows from applying \Cref{cor:avg_planted_forrelation} to $\SDist$ and $\SDist'$, via the triangle inequality (increasing $\delta$ by a factor of $2$).
\end{proof}

%\comment{W: it occurs to me that we could simplify notation by representing $\PDist_{z, L}$ as $\Dist_{z, L}$. but I'm not sure whether this would improve or worsen clarity.}

\subsection{Cryptography}

Here we define the quantum-computable primitives that are constructed in this work.

\begin{definition}
    Let $n \in \Naturals$ be the security parameter, and let $m(n)$ be the output length.
    A function $f: \{0,1\}^n \to \{0,1\}^m$ is a \emph{quantum-computable one-way function} if the following conditions hold:
    \begin{enumerate}[(i)]
        \item (Efficient computability) There is a quantum algorithm $G(x)$ running in time $\poly(n)$ that on input $x \in \{0,1\}^n$ satisfies
        \[
        \Pr\mbracket{G(x) = f(x)} \ge 2/3.
        \]
        \item (One-wayness) For all polynomial-time quantum adversaries $\mathcal{A}$,
        \[
        \Pr_{x \sim \{0,1\}^n}\mbracket{f(\mathcal{A}(1^n,f(x))) = f(x)}
        \le \negl(n).
        \]
    \end{enumerate}
    We say a (quantum-computable) one-way function is \emph{injective} if $f(\cdot)$ is injective for every $n \in \Naturals$.
\end{definition}

\begin{definition}
    Let $\kappa \in \Naturals$ be the security parameter and $n(\kappa)$ be input length.
    A keyed family of functions $\{f_k\}_{k \in \{0,1\}^\kappa}$ with $f_k: \{0,1\}^n \to \{0,1\}$ is a \emph{quantum-computable pseudorandom function family} if the following conditions hold:
    \begin{enumerate}[(i)]
        \item (Efficient computability) There is a quantum algorithm $G(k,x)$ running in time $\poly(\kappa)$ that on input $k \in \{0,1\}^\kappa$ and $x \in \{0,1\}^n$ satisfies
        \[
        \Pr\mbracket{G(k,x) = f_k(x)} \ge 2/3.
        \]
        In other words, the language $\mathcal{L}$ defined by $\mathcal{L}(k, x) = f_k(x)$ is in $\mathsf{BQP}$.
        \item (Pseudorandomness) For all polynomial-time quantum adversaries $\mathcal{A}$,
        \[
        \abs{
        \Pr_{k \sim \{0,1\}^k}\mbracket{\mathcal{A}^{f_k}(1^\kappa) = 1}
        -
        \Pr_{h \sim \{0,1\}^n \to \{0,1\}}\mbracket{\mathcal{A}^{h}(1^\kappa) = 1}
        }
        \le \negl(\kappa).
        \]
    \end{enumerate}
\end{definition}

For simplicity, in the above definition we always end up setting $\kappa = n$, although this is not strictly necessary.

\newcommand{\Gen}{\mathsf{Gen}}
\newcommand{\Eval}{\mathsf{Eval}}
\newcommand{\Inv}{\mathsf{Inv}}

\begin{definition}
\label{def:trapdoor_function}
    Let $n \in \Naturals$ be the security parameter.
    Let $\lambda(n)$, $\ell(n)$, $m(n)$ be the lengths of the public key, input, and output, respectively.
    A function $F: \{0,1\}^{\lambda} \times \{0,1\}^\ell \to \{0,1\}^m$ is a \emph{quantum-computable trapdoor one-way function} if there is triple of efficient quantum algorithms $(\Gen, \Eval, \Inv)$ that satisfies the following conditions:
    \begin{enumerate}[(i)]
        \item (Input-output behavior) $\Gen(1^n)$ outputs a pair $(pk, td) \in \{0,1\}^\lambda \times \{0,1\}^n$.\footnote{For simplicity, we take the length of the $td$ as the security parameter, although this is also without loss of generality.}
        Given $pk \in \{0,1\}^\ell$ and $x \in \{0,1\}^\ell$, $\Eval(pk, x)$ outputs $y \in \{0,1\}^m$.
        Given $td \in \{0,1\}^\ell$ and $y \in \{0,1\}^m$, $\Inv(td, y)$ outputs $x \in \{0,1\}^\ell$.
        \item (Efficient computability) For any public key $pk \in \{0,1\}^\lambda$ and input $x \in \{0, 1\}^\ell$
        \[
        \Pr\mbracket{\Eval(pk, x) = F(pk, x)} \ge 2/3.
        \]
        %\comment{i.e. $\Eval$ is in $\mathsf{BQP}$?}
        \item (Trapdoor correctness)
        For all inputs $x \in \{0,1\}^\ell$,
        \[
        \Pr_{(pk, td) \gets \Gen(1^n)}\mbracket{\Inv(td, F(pk, x)) = x} \ge 1 - \negl(n).
        \]
        \item (One-wayness) For all polynomial-time quantum adversaries $\mathcal{A}$,
        \begin{equation}
        \label{eq:trapdoor-oneway}
        \Pr_{\substack{x \sim \{0,1\}^\ell\\(pk, td) \gets \Gen(1^n)}}\mbracket{F(pk,\mathcal{A}(1^n,pk,F(pk,x))) = F(pk,x)}
        \le \negl(n).
        \end{equation}
    \end{enumerate}
    %\comment{Q: added this}
    %We further say a (quantum-computable) trapdoor one-way function is \emph{injective} if $F(pk, \cdot)$ is injective for every $pk$.
    %\comment{W: do we need injectivity anymore? TODO if not then remove this?}
\end{definition}

%\comment{Q: new stuff begins here. Is it better to move OWF and TOWF to later sections?}
Classically, a generic construction of oblivious transfer from trapdoor one-way functions is not known.
Usually, a trapdoor one-way function with some ``fakeable'' property is needed \cite{GKMRV00-pkeot} in the sense of being able to sample a $pk$ obliviously without also getting the ability to invert the sampled function.

Here we will similarly make use of the fact that our construction has pseudorandom public keys, which is a fakeable property that is useful enough for oblivious transfer.
We define it formally below.

\begin{definition}
    \label{def:towf-pseudorandom-pk}
    We say a (quantum-computable) trapdoor one-way function has \emph{pseudorandom public keys} if the distribution of $pk$ generated by $\Gen(1^n)$ is indistinguishable from $\lambda$ random bits:
    \[
    \abs{
    \Pr_{(pk,td) \leftarrow \Gen(1^n)}\mbracket{\mathcal{A}(1^n, pk) = 1}
    -
    \Pr_{pk \sim \{0,1\}^\lambda}\mbracket{\mathcal{A}(1^n, pk) = 1}
    }
    \le \negl(n).
    \]
\end{definition}

Since the challenger in the one-wayness game does not require the trapdoor once the public key is sampled, we can see that any trapdoor one-way function with pseudorandom public keys is one-way if and only if it is one-way under \emph{uniformly sampled} public keys.
Informally speaking, this implies that such a trapdoor one-way function (candidate) is one-way if and only if it is one-way in the ``fake mode'' where random public keys are given instead:

\begin{fact}
    \label{fact:fake-pk-conversion}
    Any (quantum-computable) trapdoor one-way function (candidate) with pseudorandom public keys satisfies one-wayness if and only if for all polynomial-time quantum adversaries $\mathcal A$,
        \begin{equation}
        \label{eq:trapdoor-oneway-fake}
        \Pr_{\substack{x \sim \{0,1\}^\ell\\pk \sim \{0, 1\}^\lambda}}\mbracket{F(pk,\mathcal{A}(1^n,pk,F(pk,x))) = F(pk,x)}
        \le \negl(n).
        \end{equation}
\end{fact}
\begin{proof}
    Assume for contradiction that the fact above is false, then there must exist some $\mathcal A$ that either breaks \eqref{eq:trapdoor-oneway} but does not break \eqref{eq:trapdoor-oneway-fake}, or the other way around (breaks \eqref{eq:trapdoor-oneway-fake} but not \eqref{eq:trapdoor-oneway}).
    We construct $\mathcal B(1^n, pk)$ that samples a uniformly random $x$ and outputs $1$ if and only if $F(pk,\mathcal{A}(1^n,pk,F(pk,x))) = F(pk,x)$.
    Then we see that $\mathcal B$ achieves non-negligible advantage in distinguishing the truly random distribution vs.\ the $pk$ distribution from $\Gen$.
    This contradicts the pseudorandomness of the public keys.
\end{proof}

Throughout this work, we will only consider \textit{uniform} adversaries, meaning that we have a polynomial-time Turing machine that given $1^\kappa$ outputs the quantum circuit implementing $\mathcal{A}(1^\kappa)$.
This is in contrast to the stronger notion of \textit{non-uniform} adversaries, which can be a different circuit (potentially with quantum advice) for each security parameter $\kappa$.
We believe that our lower bounds hold for non-uniform adversaries as well, though proving this could require substantially more effort---perhaps by way a \textit{direct product theorem}, that an adversary's success probability decays exponentially if it tries to solve many instances at once.
See, for example, \cite{Aar05-advice,CGLQ20-tradeoffs}.
%\comment{A: Why are we considering only uniform adversaries? Is it much harder to prove for quantum circuits? W: non-uniform lower bounds usually require direct product theorems, which are harder}

\section{Forrelation and Circuits}

We use this section to prove the sensitivity concentration properties of $\mathsf{AC^0}$ circuits that will be needed in the security proof.

For this first lemma, $\s^x(f)$ denotes the \textit{sensitivity} of a function $f: \{0,1\}^K \to \{0,1\}$ on input $x$.
It is defined as the number of single-bit changes to $x$ that change the value of $f$:
\[
\s^x(f) \coloneqq \abs{\mbrace{i \in [n]: f(x) \neq f\mparen{x^{\oplus i}}}},
\]
where $x^{\oplus i}$ is $x$ with the $i$th bit flipped.

\begin{lemma}[{\cite[Lemma 43]{AIK21-acrobatics}}]
\label{lem:ac0_sensitivity_tail_bound}
Let $f: \{0,1\}^K \to \{0,1\}$ be a circuit in $\mathsf{AC^0}[s, d]$. Then for any $t > 0$,
\[
\Pr_{x \sim \{0,1\}^K} \mbracket{ \s^x(f) \ge t } \le 2K\cdot 2^{-\frac{t}{O(\log s)^{d-1}}}.
\]
\end{lemma}

The above lemma states that $\mathsf{AC^0}$ circuits have exponentially decaying sensitivity under the uniform distribution.
Below, we prove a generalization of this statement for a ``distributional'' version of sensitivity, which measures how often the value of the function changes on a matrix-valued input when we randomly change one of the rows according to a given distribution $\Dist$.

\begin{lemma}
\label{lem:ac0_dist_block_sensitivity}
Let $f: \{0,1\}^{KM} \to \{0,1\}$ be a circuit in $\mathsf{AC^0}[s, d]$. Let $\Dist$ be a distribution over $\{0,1\}^M$. 
Let $x\sim \Dist^K$ be an input to $f$, viewed as a $K \times M$ matrix.
%Let $x \in \{0,1\}^{KM}$ be an input, viewed as a $K \times M$ matrix with $K$ rows and $M$ columns. 
Let $y$ be sampled depending on $x$ as follows: uniformly select one of the rows of $x$, resample that row from $\Dist$, and leave the other rows of $x$ unchanged. Then for any $p > 0$:
\[
\Pr_{x \sim \Dist^K}\mbracket{\Pr_y\left[f(x) \neq f(y)\right] \ge p} \le \frac{4K}{p} \cdot 2^{-\frac{p K}{O(\log s)^{d - 1}}}.
\]
\end{lemma}

\begin{proof}
    For $w \in \{0,1\}^{2KM}$, define a function $g_w: \{0,1\}^K \to \{0,1\}$ by
    \[
    g_w(z) = f(w_z),
    \]
    where we view $w$ as two $K\times M$ matrices, $w_0$ and $w_1$, and $w_z \in \{0,1\}^{K \times M}$ for $z \in \{0, 1\}^K$ is defined as follows: for each $i \in [K]$, the $i$th row of $w_z$ is taken to be the $i$th row from matrix $w_{z_i}$.

    We claim that $g_w \in \mathsf{AC^0}[s, d]$.
    This holds because each bit of $w_z$ depends on at most one bit of $z$.
    In particular,
    \[
    w_{z,i,j} = (z_i \land w_{1,i,j}) \lor (\lnot z_i \land w_{0,i,j}).
    \]
    For fixed $w$, this shows that $w_{z,i,j}$ simplifies to either $0$, $1$, $z_i$, or $\lnot z_i$. Plugging in one of these literals into the $\mathsf{AC^0}[s, d]$ circuit that computes $f$, there is no increase in either size or depth.
       
    Next, observe that for any $x \in \{0,1\}^{KM}$,
    \begin{equation}
    \label{eq:prob_vs_avg_sensitivity}
    K\cdot \Pr_y[f(x) \neq f(y)] = \E_{x' \sim \Dist^K, z \sim \{0,1\}^K}[\s^z(g_w)],
    \end{equation}
    where we take $w$ so that $w_z = x, w_{\bar{z}}= x'$.
    %\comment{Q: I rewrote this to make clear that this is definition of $w$ instead of conditioning}
    %where $\s^{z}(g_w)$ measures only the sensitivity in the $z$ variables.
	% If  $M\cdot \Pr_y[f(x) \neq f(y)]\ge t$,
	% then, by Pigeon-hole principle (since $\s^{z}(g_w)$ is surely in $[0,M]$)
	% $$
	% \Pr_{x' \sim \Dist^M, z \sim \{0,1\}^M}[\s^{z}(g_w)\ge t/2 \mid w_z = x, w_{\bar{z}}= x'] \ge \frac{t}{2M}.$$
 %    Hence, for $t=pM$, if 
 %    \begin{align*}
 %    \Pr_{x \sim \Dist^M}\left[M\cdot \Pr_y\left[f(x) \neq f(y)\right] \ge pM \right] = \delta,
 %    \end{align*}
 %    then
 %    \begin{align*}
 %    \Pr_{x \sim \Dist^M}\left[\Pr_{x' \sim \Dist^M, z \sim \{0,1\}^M}\left[\s^{z}(g_w)\ge pM/2 \mid w_z = x, w_{\bar{z}}= x'\right]\right] \ge \frac{\delta p}{2}.
 %    \end{align*}
    This will let us appeal to \Cref{lem:ac0_sensitivity_tail_bound} to complete the proof. We have:
    \begin{align*}
        \Pr_{x \sim \Dist^K}\mbracket{\Pr_y\mbracket{f(x) \neq f(y)} \ge p}
        &= \Pr_{x \sim \Dist^K}\mbracket{\E_{x' \sim \Dist^K, z \sim \{0,1\}^K}\mbracket{\s^z(g_w)} \ge pK}\\
        &\le \Pr_{x \sim \Dist^K}\mbracket{\Pr_{x' \sim \Dist^K, z \sim \{0,1\}^K}\mbracket{\s^{z}(g_w)\ge pK/2} \ge p/2}\\
        &\le \frac{2}{p}\E_{x \sim \Dist^K}\mbracket{\Pr_{x' \sim \Dist^K, z \sim \{0,1\}^K}\mbracket{\s^{z}(g_w)\ge pK/2}}\\
        &= \frac{2}{p}\Pr_{x \sim \Dist^K,x' \sim \Dist^K, z \sim \{0,1\}^K}\mbracket{\s^{z}(g_w)\ge pK/2}
    \end{align*}
    where the first line substitutes \eqref{eq:prob_vs_avg_sensitivity}, the second line uses the pigeonhole principle (since $\s^{z}(g_w)$ is surely in $[0,K]$), and the third line applies Markov's inequality.

% \comment{Q: reworded this paragraph a bit}
	In the last line, notice that sampling $z$ and then $(x,x')$, we get that $w$, where $w_z = x, w_{\bar{z}}=x'$, is independent of $z$.
 Hence, this distribution can be equivalently generated by a process that first picks $w \sim \Dist^{2K}$ and $z$ independently, which then defines $x$ and $x'$.
 So, we conclude that:

    \begin{align*}
        \Pr_{x \sim \Dist^K}\mbracket{\Pr_y\mbracket{f(x) \neq f(y)} \ge p} &\le
        \frac{2}{p}\Pr_{w \sim \Dist^{2K}, z \sim \{0,1\}^K}\mbracket{\s^{z}(g_w)\ge pK/2}\\
        &\le \frac{4K}{p}\cdot 2^{-\frac{pK}{O(\log s)^{d - 1}}},
    \end{align*}
    %\comment{A: Shouldn't it be $\frac{4K}{p}\cdot 2^{-\frac{pK}{O(\log s)^{d-1}}}$? Q: Yes thanks!}
    by substituting the exponential tail bound on $\mathsf{AC^0}$ sensitivity from \Cref{lem:ac0_sensitivity_tail_bound}.
\end{proof}

Next, we extend this lemma to the case where we indistinguishably plant a row using the patterned Forrelation distribution $\PDist_{z,L}$ for any bitstring $z$.
This lemma plays a similar role to~\cite[Lemma 46]{AIK21-acrobatics} and~\cite[Lemma 40]{KQST23-prs}, which have nearly identical proofs.

\begin{lemma}
\label{lem:patterned_forrelation_block_indistinguishable}
Fix $z \in \{0,1\}^N$. For a given $x \in \{0,1\}^{KNL}$, viewed as a $K \times NL$ matrix with $K$ rows and $NL$ columns, let $y$ be sampled depending on $x$ as follows: uniformly select one of the rows of $x$, resample that row from $\PDist_{z,L}$, and leave the other rows of $x$ unchanged.
Then for any $s = \quasipoly(L)$ and $d = O(1)$, there exists some $\delta = \frac{\polylog(L)N}{\sqrt{L}}$ such that for any circuit $f: \{0,1\}^{KNL} \to \{0,1\}$ in $\mathsf{AC^0}[s, d]$, any distribution $\SDist$ over $\{0, 1\}^N$, and any $p \ge 0$:
%\footnote{The wording here is somewhat imprecise: $\delta$ only depends on the size and depth bounds implicit in the $[\quasipoly(L), O(1)]$, and not on the circuit itself.}
\[
\Pr_{x \sim \PDist_{\SDist,L}^K}\left[\Pr_y\left[f(x) \neq f(y)\right] \ge p + \delta \right] \le \frac{4K}{p} \cdot 2^{-\frac{p K}{\polylog(L)}}.
\]
\end{lemma}

\begin{proof}
    Consider a Boolean function $C(x, w, i)$ that takes inputs $x \in \{0, 1\}^{KNL}$, $w \in \{0,1\}^{NL}$, and $i \in [K]$. Let $\tilde{y}$ be the string obtained from $x$ by replacing the $i$th row with $w$. Let $C$ output $1$ if $f(x) \neq f(\tilde{y})$, and $0$ otherwise. Observe that for any fixed $x$:
\begin{equation}
\label{eq:forrelation_block_same_distribution}
\Pr_{i \sim [K],w \sim \PDist_{z,L}}\left[C(x, w, i) = 1\right] = \Pr_{y}[f(x) \neq f(y)].
\end{equation}

We may assume without loss of generality that $N \le \sqrt{L}$, because otherwise we can pick any $\delta > 1$ to satisfy the lemma. Under this assumption, observe that for fixed $x$ and $i$, $C(x, \cdot, i) \in \mathsf{AC^0}[s+O(1), d+O(1)]$.
By \Cref{cor:avg_planted_forrelation_combined}, there exists some $\delta = \frac{\polylog(L)N}{\sqrt{L}}$ such that:
\begin{equation}
\label{eq:forrelation_block_indistinguishable}
\left|\Pr_{i \sim [K],w \sim \PDist_{z,L}}[C(x, w, i) = 1] - \Pr_{i \sim [K],w \sim \PDist_{\SDist,L}}[C(x, w, i) = 1] \right| \le \delta.
\end{equation}

Putting these together, we obtain:
\begin{align*}
\Pr_{x \sim \PDist_{\SDist,L}^K}\left[\Pr_y\left[f(x) \neq f(y)\right] \ge p + \delta \right]
&= \Pr_{x \sim \PDist_{\SDist,L}^K}\left[\Pr_{i \sim [K], w \sim \PDist_{z,L}}\left[C(x, w, i) = 1\right] \ge p + \delta \right]\\
&\le \Pr_{x \sim \PDist_{\SDist,L}^K}\left[\Pr_{i \sim [K], w \sim \PDist_{\SDist,L}}\left[C(x, w, i) = 1\right] \ge p\right]\\
&= \Pr_{x \sim \PDist_{\SDist,L}^K}\left[
\Pr_{i \sim [K],w \sim \PDist_{\SDist,L}}\left[f(x) \neq f(\tilde{y})\right]
\ge p
\right]\\
&\le \frac{4K}{p} \cdot 2^{-\frac{p K}{\polylog(L)}}.
\end{align*}
Above, the first two lines hold by \eqref{eq:forrelation_block_same_distribution} and \eqref{eq:forrelation_block_indistinguishable}, the third line holds by the definition of $C$ and $\tilde{y}$ in terms of $i$ and $w$, and the last line invokes \Cref{lem:ac0_dist_block_sensitivity} for some $s = \quasipoly(L)$ and $d = O(1)$.
\end{proof}

\section{Quantum-Computable PRF Oracle and Construction}
Here we describe the oracle that will be used to instantiate quantum-computable PRFs in a world where $\mathsf{P} = \mathsf{NP}$.
It has similarities to the oracle used in~\cite{KQST23-prs}. Like in~\cite{KQST23-prs}, the oracle consists of two parts: an oracle $A$ that is sampled from some distribution, and an oracle $B$ that is constructed to collapse $\mathsf{NP}$ to $\mathsf{P}$. The main difference is that, unlike~\cite{KQST23-prs}, $A$ is not a random oracle. Instead, $A$ is an \textit{encoding} of a random oracle, where the encoding is done using instances of the Forrelation problem.
The formal definition is below.

\begin{definition}
    \label{def:oracle_A}
    Let $\mathcal{E}_{PRF}$ be the distribution over oracles $A$ sampled via the following process:
    \begin{enumerate}[1.]
        \item For each $n \in \Naturals$ and $k \in \{0,1\}^n$, sample a uniformly random function $f_k: \{0,1\}^n \to \{0,1\}$.
        \item For each $f_k$, draw a function $\overline{f}_k : \{0,1\}^n \times \{0,1\}^{4n} \to \{0,1\}$ from $\mathcal{P}_{f_k,2^{4n}}$. In other words, for each $x \in \{0,1\}^n$, sample the truth table of $\overline{f}_k(x, \cdot)$ from the Forrelation distribution $\mathcal{F}_{2^{4n}}$ (defined in \Cref{thm:raz-tal}) if $f_k(x) = 1$, and otherwise sample it from the uniform distribution.
        %\comment{A: Might be good to remind the reader that $\mathcal{F}_{2^{4n}}$ was defined in Theorem~\ref{thm:raz-tal}.} 
        \item For each $n \in \Naturals$, $k \in \{0,1\}^n$, $x \in \{0,1\}^n$, and $y \in \{0,1\}^{4n}$, set $A(k, x, y) = \overline{f}_k(x, y)$. For all other inputs $z \in \{0,1\}^*$ not of this form, set $A(z) = 0$.
    \end{enumerate}
\end{definition}

% \begin{definition}
%     Let $A \sim \mathcal{E}$ be the following distribution over oracles $A$.
%     First, sample a uniformly random language $\mathcal{L}: \{0,1\}^* \to \{0,1\}$.
%     Then, for each $\ell \in \Naturals$ and $x \in \{0,1\}^\ell$, we simultaneously sample the values of $A(x, y)$ on all $y \in \{0,1\}^{2\ell}$ as follows:
%     \begin{enumerate}[(1)]
%         \item If $\mathcal{L}(x) = 1$, then we sample the truth table of $A(x, \cdot)$ from the Forrelation distribution $\mathcal{F}_{2^{2\ell}}$.
%         \item If $\mathcal{L}(x) = 0$, then we sample the truth table of $A(x, \cdot)$ to be a uniformly random function mapping $\{0,1\}^{2\ell} \to \{0,1\}$.
%     \end{enumerate}
% \end{definition}

The construction of $B$ from $A$ is the same as in~\cite{KQST23-prs}.

\begin{definition}[{\cite[Definition 18]{KQST23-prs}}]
\label{def:PH_oracle}
For a language $A: \{0,1\}^* \to \{0,1\}$, we define a language $\mathcal{O}[A]$ as follows. We construct an oracle $B$ inductively: for each $\ell \in\mathbb{N}$ and $x \in \{0,1\}^\ell$, view $x$ as an encoding of a pair $\langle M, y \rangle$ such that

\begin{enumerate}
\item $\langle M, y \rangle$ takes less than $\ell$ bits to specify,\footnote{Note that there are $2^\ell - 1$ such possible $\langle M, y \rangle$, which is why we take an encoding in $\{0,1\}^\ell$.}
\item $M$ is an $\mathsf{NP}$ oracle machine and $y$ is an input to $M$,
\item $M$ is syntactically restricted to run in less than $\ell$ steps, and to make queries to $A$ and $B$ on strings of length at most $\lfloor \sqrt{\ell} \rfloor$.
\end{enumerate}
Then we define $B(x) \coloneqq M(y)$. Finally, let $\mathcal{O}[A] = (A, B)$.
\end{definition}

\begin{proposition}[{\cite[Proposition 19]{KQST23-prs}}]
    For any language $A: \{0,1\}^* \to \{0,1\}$, $\mathsf{P}^{\mathcal{O}[A]} = \mathsf{NP}^{\mathcal{O}[A]}$.
\end{proposition}

The candidate quantum-computable PRF consists of the functions $\{f_k\}_{k \in \{0,1\}^n}$. We first show that this family of functions is efficiently quantum-computable:

\begin{proposition}[{cf.~\cite[Claim 30]{AIK21-acrobatics}}]
    \label{prop:prf_computable}
    %\comment{Q: Do we really care about probability $1$ now that we are no longer taking random oracles, or if existence is already enough? (It's probably a measure 0 set anyways?)}
    With probability $1$ over the random process in \Cref{def:oracle_A},
    there exists a quantum algorithm $G^A(k, x)$ running in time $\poly(n)$ such that, for every $n \in \Naturals$ and $k, x \in \{0,1\}^n$,
    \[
    \Pr\mbracket{G^A(k,x) = f_k(x)} \ge 2/3.
    \]
\end{proposition}

\begin{proof}
    For some $k, x \in \{0,1\}^n$, consider running the quantum algorithm $\mathcal{A}^{\overline{f}_k(x, \cdot)}$ from \Cref{thm:raz-tal} with $L = 2^{4n}$. Because $\overline{f}_k(x, \cdot)$ is Forrelated if $f_k(x) = 1$ and uniform if $f_k(x) = 0$, \Cref{thm:raz-tal} implies
    \[
    \Pr_{A \sim \mathcal{E}_{PRF}}\mbracket{\mathcal{A}^{\overline{f}_k(x, \cdot)} \neq f_k(x)} \le 2^{-8n},
    \]
    and hence by Markov's inequality,
    \[
    \Pr_{A \sim \mathcal{E}_{PRF}}\mbracket{\Pr\mbracket{\mathcal{A}^{\overline{f}_k(x, \cdot)} \neq f_k(x)} \ge 1/3} \le 3 \cdot 2^{-8n}.
    \]
    A union bound implies that
    \[
    \Pr_{A \sim \mathcal{E}_{PRF}}\mbracket{\exists k, x \in \{0,1\}^n : \Pr\mbracket{\mathcal{A}^{\overline{f}_k(x, \cdot)} = f_k(x)} < 2/3} \le 3 \cdot 2^{-6n}.
    \]
    Since $\sum_{n=1}^\infty 3 \cdot 2^{-6n} = 1/21 < \infty$, the Borel--Cantelli Lemma shows that with probability $1$ over $A$, the algorithm $\mathcal{A}$ satisfies
    \[
    \Pr\mbracket{\mathcal{A}^{\overline{f}_k(x, \cdot)} = f_k(x)} \ge 2/3
    \]
    for all but finitely many $(x, k)$ pairs.
    Hence, with probability $1$, there exists an algorithm $G^A(k, x)$ that satisfies the condition of the proposition on all inputs.
\end{proof}

What remains to show that $\{f_k\}_{k \in \{0,1\}^n}$ is a quantum-computable PRF is to establish its pseudorandomness against polynomial-time adversaries.

\section{PRF Security Proof}

We now have everything prepared to prove the security of the quantum-computable PRF ensemble.
Readers familiar with~\cite{KQST23-prs} will find the structure of our arguments to be quite similar.
We have made careful note of the steps below that are minor modifications of those in~\cite{KQST23-prs}.

%Take $\kappa = n$ and $L = 2^{4n}$.

% \paragraph{Challenge $\mathsf{C_1}$:} For each $k \in \{0,1\}^n$, sample $f_k: \{0,1\}^n \to \{0,1\}$ uniformly at random, and sample $\overline{f}_k : \{0,1\}^n \times \{0,1\}^{4n} \to \{0,1\}$ from $\PDist_{f_k,2^{4n}}$. 
% For $k \sim \{0,1\}^n$, choose $h = f_{k}$.

% \paragraph{Challenge $\mathsf{C_2}$:} For each $k \in \{0,1\}^n$, sample $f_k: \{0,1\}^n \to \{0,1\}$ uniformly at random, and sample $\overline{f}_k : \{0,1\}^n \times \{0,1\}^{4n} \to \{0,1\}$ from $\PDist_{f_k,2^{4n}}$. 
% Sample $h: \{0,1\}^n \to \{0,1\}$ uniformly at random.\\

To start things off, we require a quantitative version of the BBBV theorem~\cite{BBBV97-search}:

\begin{lemma}[{\cite[Lemma 37]{AIK21-acrobatics}}]
\label{lem:average_case_bbbv}
Consider a quantum algorithm $Q^x$ that makes $T$ queries to $x \in \{0,1\}^N$. Let $y \in \{0,1\}^N$ be drawn from some distribution such that, for all $i \in [N]$, $\Pr_{y}\left[x_i \neq y_i\right] \le p$. Then for any $r > 0$:
\[
\Pr_{y} \left[ \left|\Pr\left[Q^y = 1 \right] - \Pr\left[Q^x = 1 \right] \right| \ge r \right] \le \frac{64pT^2}{r^2}.
\]
\end{lemma}

Next we introduce a bit more notation that was used in~\cite{KQST23-prs}.
For $A: \{0,1\}^* \to \{0,1\}$ an oracle, let $B$ be defined depending on $A$ as in \Cref{def:PH_oracle}. For any $t \in \Naturals$, denote by $A_{\le t}$ (respectively, $B_{\le t}$) the concatenation of $A(x)$ (respectively, $B(x)$) over all strings $x$ of length at most $t$. The next lemma uses the standard connection between $\mathsf{PH}$ algorithms and $\mathsf{AC^0}$ circuits~\cite{FSS84-circuit-oracle} to show that each bit of an oracle constructed as $B$ can be computed by a small $\mathsf{AC^0}$ circuit in the bits of $A$. Think of $d$ as some constant.

\begin{lemma}[{\cite[Lemma 35]{AIK21-acrobatics}}]
\label{lem:PH_to_AC0}
Fix $t, d \in \Naturals$, and let $t' \le t^{2^d}$. For each $x \in \{0,1\}^{t'}$, there exists an $\mathsf{AC^0}\mbracket{O\mparen{2^{t^{2^d}}}, 2d}$ circuit that takes as input $A_{\le t^{2^{d-1}}}$ and $B_{\le t}$ and computes $B(x)$.
\end{lemma}

We now prove a lemma, which informally states that we can plant a random Forrelation-encoded block at a random location indistinguishably against $\mathsf{BQP^{PH}}$, even if the adversary is given all the information about the sampled (unencoded) random block that we are planting.
Its proof resembles the proof of~\cite[Theorem 28]{KQST23-prs}.
(In order to understand the lemma conceptually, it is best to ignore the quantitative bounds on $N$ and $L$ in terms of $n$, other than to note that they should both be roughly $2^{\poly(n)}$.)
%Furthermore, this is true even if we give the adversary the unchanged oracle $A$ at the side.
\comment{Avishay's suggestion: move this after Theorem 21?}
%\comment{A:Why do we credit KQST for this lemma? W: just to show that the proof is similar.}
\begin{lemma}
\label{lem:pseudorandom_image}
    %\comment{Q: new lemma for Forrelation resampling extracted from the theorem below, and generalized for arbitrary distributions}
    Let %$N = 2^{\poly(n)}$ and
    $N^2 \cdot 2^{2n} \le L \le 2^{\poly(n)}$.
    %\comment{A: I think you want $L \le 2^{\poly(n)}$ as well. Q: good catch}
    Consider an oracle distribution $A \sim \mathcal{E}$ for which the following is true: for every integer $n > 0$, there exists a region of the oracle of size $2^n \times NL$ bits such that, conditioned on the rest of the oracle $A$, the region is identically distributed as $\PDist_{\SDist, L}^{2^n}$ for some $\SDist$.\footnote{To clarify, $\SDist$ may depend on the portion of $A$ outside of this $2^n \times NL$ region.}
    %\comment{W: Does this mean $\mathcal{S}$ can depend on $A$? Q: Yes technically both the distribution of $k$ and $\SDist$ can depend on $A$ outside of the region in interest, in other words, this should be true even if the entire rest of the oracle is fixed}.
    Then for all $\poly(n)$-query\footnote{Here, we assume a $\poly(n)$ upper bound on both the number of queries and the length of the longest query to the oracle.
    We ignore the time to read $z$, which may be exponentially long.} quantum algorithms $\mathcal A$,
    \[
    \E_{\substack{A \sim \mathcal{E}\\z \sim \SDist}}\abs{
    \Pr_{\substack{k \sim [{2^n}]\\w \sim \PDist_{z, L}}}\mbracket{\mathcal{A}^{\mathcal{O}[A_{k \mapsto w}]}(z) = 1}
    -
    \Pr\mbracket{\mathcal{A}^{\mathcal{O}[A]}(z) = 1}
    } = \negl(n),
    \]
    where $A_{k \mapsto w}$ is the same as $A$ but replacing the $k$th block of $NL$ bits with $w$.
\end{lemma}
\begin{proof}
Let $T = \poly(n)$ be an upper bound on the number of queries that $\mathcal{A}$ makes to $\mathcal{O}[A]$, and also on the length of strings queried in $\mathcal{O}[A]$. For some $p$ that we choose later, call $(A, z)$ ``good'' if, for all $v \in \{0,1\}^{\le T}$, $\Pr_{k, w}\left[\mathcal{O}[A](v) \neq \mathcal{O}[A_{k \mapsto w}](v)\right] \le p$.
We have that for any $r \in [0, 1]$,

\begin{align}
\text{LHS} &= \E_{\substack{A \sim \mathcal{E}\\z \sim \SDist}}\abs{
\Pr_{\substack{k \sim [{2^n}]\\w \sim \PDist_{z, L}}}\mbracket{\mathcal{A}^{\mathcal{O}[A_{k \mapsto w}]}(z) = 1}
-
\Pr\mbracket{\mathcal{A}^{\mathcal{O}[A]}(z) = 1}
} \nonumber \\
&\le \E_{\substack{A \sim \mathcal{E}\\z \sim \SDist}}\left[ r + 
\Pr_{\substack{k \sim [{2^n}]\\w \sim \PDist_{z, L}}}\left[\left|\Pr\left[\mathcal{A}^{\mathcal{O}[A_{k \mapsto w}]}(z) = 1 \right] - \Pr\left[\mathcal{A}^{\mathcal{O}[A]}(z) = 1 \right]\right| \ge r\right] \right]\nonumber\\
&\le r + \frac{64pT^2}{r^2} + \Pr_{\substack{A \sim \mathcal{E}\\z \sim \SDist}}[(A, z) \text{ not good}],\label{eq:p2_p3_Az_not_good}
\end{align}
where in the second line we have applied \Cref{lem:average_case_bbbv} when $(A, z)$ is good.

We turn to upper bounding the probability that $(A, z)$ is not good.
By \Cref{lem:PH_to_AC0} with $t = n$, $d = \lceil\log_2 \log_n T\rceil = O(1)$ and $t' = T$, for every $v \in \{0, 1\}^{\le T}$, $\mathcal{O}[A](v)$ can be computed by an $\mathsf{AC^0}\mbracket{O\mparen{2^{T^2}}, O(1)}$ circuit $C$ whose inputs are $A_{\le T}$ and $B_{\le n}$.
Because $B_{\le n}$ by construction does not depend at all on $A_{> n}$ (see \Cref{def:PH_oracle}), the only inputs to this circuit that can change between $A$ and $A_{k \mapsto w}$ are those corresponding to the region of size $2^n \times NL$.
Using the notation of \Cref{lem:patterned_forrelation_block_indistinguishable} with $K = 2^n$, we have that for some $\delta = \frac{\polylog(L)N}{\sqrt{L}} = \frac{\poly(n)}{2^n}$,
\begin{align*}
    \Pr_{\substack{A \sim \mathcal{E}\\z \sim \SDist}}\left[\Pr_{\substack{k \sim [{2^n}]\\w \sim \PDist_{z, L}}}\left[\mathcal{O}[A](v) \neq \mathcal{O}[A_{k \mapsto w}](v)\right] \ge q + \delta\right] &= \Pr_{x \sim \PDist_{\SDist,L}^K}\left[\Pr_{y}\left[C(x) \neq C(y)\right] \ge q + \delta\right]\\
    &\le \frac{4K}{q} \cdot 2^{-\frac{q K}{\polylog(L)}}\\
    &\le \frac{1}{q} \cdot 2^{n+2-\frac{q 2^n}{\poly(n)}}.
\end{align*}
%where the second line holds by \Cref{lem:patterned_forrelation_block_indistinguishable}, the third line substitutes $M$ and $N$, and the fourth line uses the assumption that $n \le \poly(\kappa)$. 
Choose $q = 2^{-n/2}$ and $p = q + \delta$. Then by a union bound over all $v \in \{0,1\}^{\le T}$, we conclude that
\begin{align*}
    \Pr_{\substack{A \sim \mathcal{E}\\z \sim \SDist}}[(A, z) \text{ not good}] &= \Pr_{A,z}\left[ \exists v \in \{0,1\}^{\le T} : \Pr_{k, w}\left[\mathcal{O}[A](v) \neq \mathcal{O}[A_{k \mapsto w}](v)\right] > p\right]\\
    &\le \sum_{v \in \{0,1\}^{\le T}}\Pr_{A,z}\left[ \Pr_{k, w}\left[\mathcal{O}[A](v) \neq \mathcal{O}[A_{k \mapsto w}](v)\right] \ge p\right]\\
    &\le 2^{T + 1} \cdot \frac{1}{q} \cdot 2^{n+2-\frac{q 2^n}{\poly(n)}}\\
    &\le 2^{T+3n/2+2-\frac{2^{n/2}}{\poly(n)}}\\
    &\le 2^{\poly(n)-\frac{2^{n/2}}{\poly(n)}}\\
    &\le 2^{-2^{\Omega(n)}},
\end{align*}
where in the penultimate line we substituted $T = \poly(n)$.
Combining with \eqref{eq:p2_p3_Az_not_good} and choosing $r = 2^{-n / 6}$ gives us the final bound
\begin{align*}
    \text{LHS}
    &\le r + \frac{64pT^2}{r^2} + 2^{-2^{\Omega(n)}}\\
    &\le r + \frac{64T^2}{r^2}\left(2^{-n/2} + \frac{\poly(n)}{2^n} \right) + 2^{-2^{\Omega(n)}}\\
    &\le \frac{\poly(n)}{2^{n / 6}}\\
    &= \negl(n),
\end{align*}
because again $T = \poly(n)$.
\end{proof}

In plain words, the theorem below shows that the advantage of an efficient adversary, averaged over the distribution of oracles, is small.
The proof mirrors~\cite[Theorem 28]{KQST23-prs}.

\begin{theorem}.
\label{thm:security_no_abs}
    For all polynomial-time quantum adversaries $\mathcal{A}$, 
    \[
    \abs{\E_{A \sim \mathcal{E}_{PRF}}\mbracket{
    \Pr_{k \sim \{0,1\}^n}\mbracket{\mathcal{A}^{\mathcal{O}[A],f_{k}}(1^n) = 1}
    -
    \Pr_{h \sim \{0,1\}^n \to \{0,1\}}\mbracket{\mathcal{A}^{\mathcal{O}[A],h}(1^n) = 1}
    }}
    = \negl(n).
    \]
\end{theorem}

\begin{proof}
We will show that this security property depends only on the query bound that the adversary makes to $\mathcal{O}[A]$: in particular, the advantage remains negligible even if $\mathcal{A}$ can access the entire truth table of either $f_{k}$ or $h$, which we denote by $\mathcal{A}^{\mathcal{O}[A]}(1^n,f_{k})$ or $\mathcal{A}^{\mathcal{O}[A]}(1^n,h)$.

We first use a similar technique to~\cite[Theorem 28]{KQST23-prs}, where we show how to switch between the two security challenges by changing a small part of the oracle.
%For $A: \{0,1\}^* \to \{0,1\}$ and $h: \{0,1\}^n \to \{0,1\}$, let $A^* \sim \mathcal{G}_{A, h}$ mean that we take $A^*$ to be an ``adjacent'' random sample from the pseudorandom security challenge that is consistent with $h = f_{k}$ for some $k \in \{0,1\}^n$. We define this by taking $A^*$ to be identical to $A$, except that for a uniformly random $k \in \{0,1\}^n$ we take $\overline{f}_k \sim \PDist_{h,2^{4n}}$. Observe that sampling $A \sim \mathcal{E}$ and $h \sim \{0,1\}^n \to \{0,1\}$ followed by $A^* \sim \mathcal{G}_{A, h}$ is equivalent to sampling $(A^*, h)$ from the pseudorandom security challenge. So, we have:
Suppose we have sampled $A \sim \mathcal{E}_{PRF}$ and $h \sim \{0,1\}^n \to \{0,1\}$, consistent with a uniformly random challenge.
Now, consider drawing $k \sim \{0,1\}^n$, $w \sim \PDist_{h, L}$, and replacing the oracle with $A_{k \mapsto w}$, which means replacing the region of $A$ indexed by $k$ with $w$.
Observe that the resulting $(A_{k \mapsto w}, h)$ is equivalent to a draw from the pseudorandom security challenge. So, we have:

\begin{align}
&\equad\abs{\E_{A \sim \mathcal{E}_{PRF}}\mbracket{
    \Pr_{k \sim \{0,1\}^n}\mbracket{\mathcal{A}^{\mathcal{O}[A]}(1^n,f_{k}) = 1}
    -
    \Pr_{h \sim \{0,1\}^n \to \{0,1\}}\mbracket{\mathcal{A}^{\mathcal{O}[A]}(1^n,h) = 1}
    }}\nonumber\\
&= \abs{\E_{\substack{A \sim \mathcal{E}_{PRF}\\h \sim \{0,1\}^n \to \{0,1\}}}\mbracket{
\Pr_{\substack{k \sim \{0,1\}^n\\w \sim \PDist_{h,L}}} \left[\mathcal{A}^{\mathcal{O}[A_{k \mapsto w}]}(1^n, h) = 1 \right] - \Pr\left[\mathcal{A}^{\mathcal{O}[A]}(1^n, h) = 1 \right] }}\nonumber\\
&\le \E_{\substack{A \sim \mathcal{E}_{PRF}\\h \sim \{0,1\}^n \to \{0,1\}}}\abs{
\Pr_{\substack{k \sim \{0,1\}^n\\w \sim \PDist_{h,L}}} \left[\mathcal{A}^{\mathcal{O}[A_{k \mapsto w}]}(1^n, h) = 1 \right] - \Pr\left[\mathcal{A}^{\mathcal{O}[A]}(1^n, h) = 1 \right] }.\label{eq:p2_p3_A*}
\end{align}
We note that the last quantity is negligible from \Cref{lem:pseudorandom_image}, by considering $h = z$, $N = 2^n$, $L = 2^{4n}$, and identifying $\{0,1\}^n$ with $[2^n]$.
\end{proof}

By way of standard techniques, we turn \Cref{thm:security_no_abs} into a proof that the quantum-computable PRF is secure. First we show that we can replace absolute value in \Cref{thm:security_no_abs} inside of the expectation, following~\cite[Corollary 32]{KQST23-prs}:

\begin{corollary}
\label{cor:security_with_abs}
    For all polynomial-time quantum adversaries $\mathcal{A}$, 
    \[
    \E_{A \sim \mathcal{E}_{PRF}}\abs{
    \Pr_{k \sim \{0,1\}^n}\mbracket{\mathcal{A}^{\mathcal{O}[A],f_{k}}(1^n) = 1}
    -
    \Pr_{h \sim \{0,1\}^n \to \{0,1\}}\mbracket{\mathcal{A}^{\mathcal{O}[A],h}(1^n) = 1}
    }
    = \negl(n).
    \]
\end{corollary}

\begin{proof}
    Suppose $\mathcal{A}$ outputs a single bit. For a fixed oracle $A$, define the probabilities
    \[
    a(A) \coloneqq \Pr_{k \sim \{0,1\}^n}\mbracket{\mathcal{A}^{\mathcal{O}[A],f_{k}}(1^n) = 1}.
    \]
    and
    \[
    b(A) \coloneqq \Pr_{h \sim \{0,1\}^n \to \{0,1\}}\mbracket{\mathcal{A}^{\mathcal{O}[A],h}(1^n) = 1}.
    \]
    We wish to bound $\E_{A \sim \mathcal{E}_{PRF}}\mbracket{a(A) - b(A)}$. Define an adversary $\mathcal{B}^{\mathcal{O}[A], \cdot}$ %\comment{A: This should be : built from $\mathcal{A}...$ right? Q: fixed}
    built from $\mathcal{A}^{\mathcal{O}[A], \cdot}$ that executes as follows:
    \begin{enumerate}[1.]
        \item Sample $c \sim \{0,1\}$.
        \item If $c = 1$, then sample $k \sim \{0,1\}^n$ and run $\mathcal{A}^{\mathcal{O}[A],f_{k}}(1^n)$. If $c = 0$ then simulate $\mathcal{A}^{\mathcal{O}[A],h}(1^n)$ on a uniformly random $h: \{0,1\}^n \to \{0,1\}$ using Zhandry's compressed oracle technique~\cite{Zha19-compress}. Either way, call the output $d$.
        \item Output $c \oplus d \oplus \mathcal{A}^{\mathcal{O}[A], \cdot}(1^n)$.
    \end{enumerate}

    The key observation is that $c \oplus d = 0$ with probability $\frac{1 + a(A) - b(A)}{2}$ and $c \oplus d = 1$ with probability $\frac{1 - a(A) + b(A)}{2}$. Hence, the distinguishing advantage of $\mathcal{B}$ satisfies:
    \begin{align*}
    &\E_{A \sim \mathcal{E}_{PRF}}\mbracket{
    \Pr_{k \sim \{0,1\}^n}\mbracket{\mathcal{B}^{\mathcal{O}[A],f_{k}}(1^n) = 1}
    -
    \Pr_{h \sim \{0,1\}^n \to \{0,1\}}\mbracket{\mathcal{B}^{\mathcal{O}[A],h}(1^n) = 1}
    }\\
    &\qquad\qquad = \E_{A \sim \mathcal{E}_{PRF}}\mbracket{\Pr[c \oplus d = 0](a(A) - b(A)) + \Pr[c \oplus d = 1](b(A) - a(A))}\\
    &\qquad\qquad = \E_{A \sim \mathcal{E}_{PRF}}\mbracket{(a(A) - b(A))^2}.
    \end{align*}
    The desired quantity is then bounded by:
    \begin{align*}
        \E_{A \sim \mathcal{E}_{PRF}}\mbracket{a(A) - b(A)}
        &\le
        \sqrt{\E_{A \sim \mathcal{E}_{PRF}}[(a(A) - b(A))^2]}\\
        &= \sqrt{\E_{A \sim \mathcal{E}_{PRF}}\mbracket{
    \Pr_{k \sim \{0,1\}^n}\mbracket{\mathcal{B}^{\mathcal{O}[A],f_{k}}(1^n) = 1}
    -
    \Pr_{h \sim \{0,1\}^n \to \{0,1\}}\mbracket{\mathcal{B}^{\mathcal{O}[A],h}(1^n) = 1}
    }}\\
    &\le \sqrt{\negl(n)}\\
    &\le \negl(n),
    \end{align*}
    by Jensen's inequality and \Cref{thm:security_no_abs}.
\end{proof}

\begin{theorem}
\label{thm:main_prf}
    With probability $1$ over $A \sim \mathcal{E}_{PRF}$, $\{f_k\}_{k \in \{0,1\}^n}$ is a quantum-computable pseudorandom function family relative to $\mathcal{O}[A]$.
\end{theorem}

\begin{proof}
    The efficient computability of the ensemble was established in \Cref{prop:prf_computable}, so we turn to proving its pseudorandomness against polynomial-time adversaries.

    \Cref{cor:security_with_abs} shows that for any polynomial-time adversary $\mathcal{A}$, there is a negligible function $\delta(n)$ such that:
    \[
    \E_{A \sim \mathcal{E}_{PRF}}\abs{
    \Pr_{k \sim \{0,1\}^n}\mbracket{\mathcal{A}^{\mathcal{O}[A],f_{k}}(1^n) = 1}
    -
    \Pr_{h \sim \{0,1\}^n \to \{0,1\}}\mbracket{\mathcal{A}^{\mathcal{O}[A],h}(1^n) = 1}
    }
    \le \delta(n).
    \]
    Markov's inequality then implies
    \[
    \Pr_{A \sim \mathcal{E}_{PRF}}\mbracket{
    \abs{
    \Pr_{k \sim \{0,1\}^n}\mbracket{\mathcal{A}^{\mathcal{O}[A],f_{k}}(1^n) = 1}
    -
    \Pr_{h \sim \{0,1\}^n \to \{0,1\}}\mbracket{\mathcal{A}^{\mathcal{O}[A],h}(1^n) = 1}
    }
    \ge \sqrt{\delta(n)}
    }
    \le \sqrt{\delta(n)}.
    \]
%    \comment{Q: slightly changed this. negligible functions can have e.g.\ $\delta(1) = 1$}
    Since $\delta$ is negligible, there exists $n_0$ such that $\delta(n) \le n^{-4}$ for all $n \ge n_0$.
    Then infinite sum $\sum_{n=n_0}^\infty \sqrt{\delta(n)} \le \sum_{n = 0}^\infty n^{-2} < \infty$, so the Borel--Cantelli Lemma implies that with probability $1$ over $A \sim \mathcal{E}_{PRF}$,
    \[
    \abs{
    \Pr_{k \sim \{0,1\}^n}\mbracket{\mathcal{A}^{\mathcal{O}[A],f_{k}}(1^n) = 1}
    -
    \Pr_{h \sim \{0,1\}^n \to \{0,1\}}\mbracket{\mathcal{A}^{\mathcal{O}[A],h}(1^n) = 1}
    }
    \le \sqrt{\delta(n)}
    \]
    except for finitely many $n$'s.
    This shows that the distinguishing advantage is bounded by a negligible function with probability $1$. Because there are only countably many uniform polynomial-time quantum adversaries, we can union bound over all of them to conclude that with probability $1$, \textit{every} adversary $\mathcal{A}$ has at most a negligible distinguishing advantage.
\end{proof}

We conclude by observing that \Cref{thm:main_prf} gives rise to a construction of quantum-computable injective one-way functions as well.
The proof uses a standard construction of a one-way function from a random oracle (see, e.g.\ \cite{IR89-permutations}), and injectivity follows from the fact that an expanding random function is with overwhelming probability injective.
%Technically, we even show something slightly stronger: that there also exist quantum-computable \textit{injective} one-way functions relative to this oracle (\Cref{cor:injective_qcowf}).
While one-way functions are well-known to be equivalent to pseudorandom functions and many other ``Minicrypt'' primitives, \textit{injective} one-way functions are somewhat more powerful.
For example, they suffice to build non-interactive commitment schemes with classical communication~\cite{Blu89-coin,GL89-hard-core,MP12-nic}.

%\comment{Is this reduction similar/inspired by something in the literature? If so, we should give it credit. Q: not sure if it is in the literature but it is pretty standard. W: added~\cite{IR89-permutations} Q: will double check offline}
\begin{corollary}
    \label{cor:injective_qcowf}
    With probability $1$ over $A \sim \mathcal{E}_{PRF}$, there exists a quantum-computable injective one-way function relative to $\mathcal{O}[A]$.
\end{corollary}

\begin{proof}
    Let $\{f_k\}$ be the quantum-computable pseudorandom function family from \Cref{thm:main_prf}.
    Consider the function $g: \{0,1\}^n \to \{0,1\}^{3n}$ defined by:
    \[
    g(k) = (f_k(1),f_k(2),\ldots,f_k(3n)),
    \]
    where the numbers $1,\ldots,3n$ are treated as their binary representations in $\{0,1\}^n$ (assuming $n \ge 4$).

    We claim that $g$ is a quantum-computable one-way function.
    $g$ is clearly efficiently quantum computable, by way of the algorithm for evaluating $f_k$, after suitably applying error reduction to compute all $3n$ output bits with high probability.
    The one-wayness of $g$ follows from the pseudorandomness of $\{f_k\}$ established in \Cref{thm:main_prf}, for if there were an efficient inversion algorithm to find $k$ given $(f_k(1),f_k(2),\ldots,f_k(3n))$ with noticeable probability, one could use this to distinguish the pseudorandom function family from a random function.
    In particular, given oracle access to $h: \{0,1\}^n \to \{0,1\}$ which is either uniformly random or pseudorandom (i.e., $h=f_k$ for a random $k$), the adversary would query $(h(1),h(2),\ldots,h(3n))$, run the inversion algorithm to obtain $k$, and then accept if and only if $h(i) = f_k(i)$ for all $i \in [3n]$.
    Assuming the inversion algorithm succeeds with non-negligible probability, then on $h = f_k$ a pseudorandom challenge, the adversary accepts with non-negligible probability.
    By contrast, the adversary rejects with high probability on a uniformly random $h$, because the probability that there exists an $f_k$ consistent with $h$ on these $3n$ inputs is at most
    \[
    \sum_{k \in \{0,1\}^n} \Pr[\forall i \in [3n]: h(i) = f_k(i) ] = 2^n \cdot 2^{-3n} = 2^{-2n}.
    \]
    %\comment{A: Can we spell it out more slowly? I think it is not that clear what we are given and what we do with it. Perhaps we can say: Given access to a function $f$ which is either drawn from $f_k$ or at random, we can tell the difference between the two cases as follows. We generate the list $(f(1), \ldots, f(3n))$ use the inverter to get a parameter $k$ and check that $(f_k(1), \ldots, f_k(3n)) = (f(1), \ldots, f(3n))$. Indeed, for $f$ drawn from $f_k$ the inverter would work with non-negligible probability but for $f$ being a random function, the probability that $(f(1), \ldots, f(3n))$ is in the image of $g$ is at most $2^{-2n}$ which is negligible.}

    It remains to establish injectivity.
    We claim that with probability $1$, $g$ is injective on all but finitely many input lengths.
    For a particular input length $n \in \Naturals$, the probability that $g$ fails to be injective is at most
    \[
    \sum_{x\neq y \in \{0,1\}^n} \Pr[g(x) = g(y)] \le \binom{2^n}{2} 2^{-3n} \le 2^{-n}.
    \]
    So, by the Borel--Cantelli Lemma, because $\sum_{n=1}^\infty 2^{-n} = 1 < \infty$, $g$ is injective on all but finitely many input lengths with probability $1$.
    Therefore, $g$ can be modified on finitely many inputs into a different function $g'$ that is injective on \textit{all} input lengths.
    This $g'$ remains a quantum-computable one-way function because it differs from $g$ in only finitely many places.
\end{proof}

\iftrapdoor
\section{Trapdoor Function Oracle and Construction}
%\comment{A: Maybe: Subsection for Construction and Subsection for Security Proof. Same for Sections 4+5}

In this section, we build the oracle relative to which quantum-computable trapdoor one-way functions exist but $\mathsf P = \mathsf{NP}$. 
(We defer the security proof to the next section.)

The oracle still consists of $A, B$ where $B$ is recursively constructed from $A$, and $A$ is a Forrelation-encoding as before.
Instead of $A$ encoding a random oracle, now $A$ will be an encoding of three functions $G, F, I$ each parameterized (implicitly, for notational convenience) by a security parameter $n \in \Naturals$.
We first describe the unencoded functions:
%\comment{W TODO: say something about reconciling w/ the notation we used for the random oracle}
\begin{itemize}
    \item $G: \{0, 1\}^n \to \{0, 1\}^{3n}$ is a random function mapping a trapdoor $td$ to a public key $pk$.
    \item $F: \{0, 1\}^{3n} \times \{0, 1\}^n \to \{0, 1\}^{6n}$ is a random function mapping a public key $pk$ and a pre-image $x$ to an image $y$.
    \item $I: \{0, 1\}^n \times \{0, 1\}^{6n} \to \{0, 1\}^n \cup \{\bot\}$ is the inversion function determined by $F$ and $G$.
        In particular, $I(td, y) = F_{G(td)}^{-1}(y)$ where $F_{pk}^{-1}(y)$ is the lexicographically smallest bitstring $x$ that satisfies $F(pk, x) = y$, or a special symbol $\bot$ if such a bitstring does not exist.
\end{itemize}

The encoding $\overline F$ of $F$ takes as input $pk, x, i, z$ where $\overline F(pk, x, i, \cdot)$ is a Forrelated instance of length $L = 2^{15n}$ if the $i$th bit of $F(pk, x)$ is $1$, or uniformly random otherwise.
The encoding works similarly for $\overline G, \overline I$.
(We treat the output of $I$ as an $(n + 1)$-bit string.)
Finally, as in \Cref{def:oracle_A}, we choose a suitable encoding so that all of $\overline{G}$, $\overline{F}$, and $\overline{I}$ are accessible by querying $A$.
We henceforth denote this distribution over oracles by $A \sim \mathcal{E}_{TD}$.
%\comment{W: maybe use a different letter here? Q: TODO change to $\mathcal{E}_{PRF}$ vs $\mathcal{E}_{TD}$. W: done}

% \begin{definition}
%     Let $n \in \Naturals$ be the security parameter, and let $\ell(n), m(n)$ be the input and output length respectively.
%     A \emph{quantum-computable one-way function} is a triple of efficient quantum algorithms $(\Gen, f, f^{-1})$ \comment{W: maybe call them $\Gen$, $\Eval$, $\Inv$?} that satisfies the following conditions:
%     \begin{enumerate}[(i)]
%         \item (Efficient computability) For any input $pk \in \{0,1\}^n, x \in \{0, 1\}^\ell$\comment{W: above $pk$ has length $3n$ and $x$ has length $n$...}, there exists some $y \in \{0, 1\}^m$ such that
%         \[
%         \Pr\mbracket{f(pk, x) = y} \ge 2/3.
%         \]
%         For the rest of the presentation, we use $f(pk, x)$ to instead denote the output value $y$ unless otherwise stated.
%         This is because we can always perform majority vote to boost the correctness error to $2^{-n}$ with a polynomial-time slowdown.
%         \item (Trapdoor correctness)
%         \comment{W: is this for all $x$ or whp over random $x$?}
%         \[
%         \Pr\mbracket{(pk, td) \gets \Gen(1^n); f^{-1}(td, f(pk, x)) = x} = 1 - \negl(n).
%         \]
%         \item (One-wayness) For all polynomial-time quantum adversaries $\mathcal{A}$,
%         \[
%         \Pr_{x \sim \{0,1\}^\ell}\mbracket{(pk, td) \gets \Gen(1^n); f(pk,\mathcal{A}(1^n,pk,f(pk,x))) = f(pk,x)}
%         = \negl(n).
%         \]
%     \end{enumerate}
% \end{definition}

We wish to show that computing $(G, F, I)$ gives a quantum-computable trapdoor one-way function.
We first show that it satisfies the conditions of \Cref{def:trapdoor_function} other than one-wayness, using the same Forrelation decoding argument as \Cref{prop:prf_computable}.

\begin{proposition}
    \label{prop:towf-correctness}
    With probability $1$ over $A \sim \mathcal{E}_{TD}$, there exists a triple of algorithms $(\Gen, \Eval, \Inv)$ satisfying the input-output behavior, efficient computability, trapdoor correctness criteria of \Cref{def:trapdoor_function}.
\end{proposition}

\begin{proof}
    Using the same approach as \Cref{prop:prf_computable}, all of $G, F, I$ can be evaluated (with probability $1$, on all but finitely many inputs) by a bounded-error quantum that queries the Forrelation encodings $\overline{G}, \overline{F}, \overline{I}$.
    The algorithm for $\Gen(1^n)$ is to sample a random $td \in \{0,1\}^n$ and evaluate $(G(td), td)$ using the bounded-error quantum algorithm mentioned above.
    Similarly, $\Eval(pk, x)$ evaluates $F(pk, x)$, and $\Inv(td, y)$ evaluates $I(td, y)$ with bounded-error.
    In the case of $\Gen$ and $\Inv$, we amplify the success probability $2/3$ to $1 - 2^{-n}$.
    These algorithms clearly satisfy the input-output behavior and efficient computability criteria with $\lambda = 3n$, $\ell = n$, and $m = 6n$, so it remains to establish trapdoor correctness. For this we show the following claim:

    \begin{claim}
        \label{claim:f_injective}
        With probability $1$ over $A \sim \mathcal{E}_{TD}$, for all but finitely many public keys $pk \in \{0,1\}^*$, $F(pk, \cdot)$ is injective.
    \end{claim}

    For a fixed $pk \in \{0,1\}^{3n}$, the probability that two distinct $x, x' \in \{0,1\}^n$ satisfy $F(pk, x) = F(pk, x')$ is at most $2^{2n} \cdot 2^{-6n} = 2^{-4n}$, by the union bound.
    So, for a given input length $n$, the probability that there exists a $pk \in \{0,1\}^{3n}$ for which $F(pk, \cdot)$ is not injective is at most $2^{-n}$.
    By the Borel-Cantelli lemma, since $\sum_{n=1}^\infty  2^{-n} = 1 < \infty$, the claim is established.

    Now observe that if $F(pk, \cdot)$ is injective for all $pk \in \{0,1\}^{3n}$, then for any input $x \in \{0,1\}^n$ we have:
    % \begin{align*}
    % \Pr_{(pk, td) \gets \Gen(1^n)}\mbracket{\Inv(td, F(pk, x)) = x} &\ge \Pr_{td \sim \{0,1\}^n}\mbracket{I(td, F(G(td), x)) = x} - 2\cdot 2^{-n}\\
    % &= 1 - 2 \cdot 2^{-n}\\
    % &\ge 1 - \negl(n),
    % \end{align*}
    % \comment{A: I couldn't verify the above equations. What's the error of Gen?}
    % because $\Gen$ and $\Inv$ err with probability at most $2^{-n}$ each.
    % The lemma now follows from \Cref{claim:f_injective}.
    % (The $1 - 2 \cdot 2^{-n}$ probability bound may change on the finitely many inputs for which \Cref{claim:f_injective} fails, but it is still $1 - \negl(n)$.)
    \comment{A: I changed it a bit. W: looks good.}
    \begin{align*}
    \Pr_{(pk, td) \gets \Gen(1^n)}\mbracket{\Inv(td, F(pk, x)) = x} &\ge 1 - 2 \cdot 2^{-n}.
    \end{align*}
    This is because (i) $\Gen(1^n)$ outputs a pair $(pk, td)$ where $pk = G(td)$ with probability at least $1-2^{-n}$, (ii) $\Inv$ executes $I$ correctly with probability at least $1-2^{-n}$, and (iii) when both (i) and (ii) happen and $F(pk, \cdot)$ is injective, then $\Inv(td, F(pk, x)) = I(td, F(G(td),x)) = x$.
    %
    % \comment{A: I couldn't verify the above equations. What's the error of Gen?}
    % because $\Gen$ and $\Inv$ err with probability at most $2^{-n}$ each.
    % The lemma now follows from \Cref{claim:f_injective}.
    % (The $1 - 2 \cdot 2^{-n}$ probability bound may change on the finitely many inputs for which \Cref{claim:f_injective} fails, but it is still $1 - \negl(n)$.)
\end{proof}

%\comment{In fact, it follows from the proof above that $(F, I)$ can be modified on finitely many points so that it is injective for all $pk$ while preserving all other properties. For simplicity, we omit this.}

It remains to show that $F$ is hard to invert against $\mathsf{BQP^{PH}}$ in the next section.

%\comment{TODO do we need $F$ to be injective, and if so what does that mean? we sorta establish it in the claim above anyways... Q: added definition. Trapdoor correctness implies injective on average, and we can get worse-case injectiveness which might be useful but I'm not sure}

\section{Trapdoor Function Security Proof}
\label{sec:trapdoor-security}

We now turn to proving one-wayness.
%\comment{Q: the following can be moved to tech overview}
On a high level, the proof is going to break down into three steps:
\begin{enumerate}
    \item Prove that our construction has pseudorandom public keys (\Cref{def:towf-pseudorandom-pk}), which again means that the public key distribution generated by $\Gen$ is pseudorandom.
        The proof idea is roughly the same as \Cref{thm:main_prf} but we instead perform a careful resampling of $G$ and $I$ at the same time.
    \item Prove that our construction is one-way against the uniformly random public key distribution (see \eqref{eq:trapdoor-oneway-fake}).
        The proof idea is that with overwhelming probability that we will actually sample a public key without any corresponding inversion oracles, thus one-wayness reduces to pseudorandomness (similar to \Cref{cor:injective_qcowf}), which can be again established via a similar argument as \Cref{thm:main_prf}.
    \item Combine the two above and conclude that it is one-way by \Cref{fact:fake-pk-conversion}.
\end{enumerate}
%Our task is to show that an adversary, given $pk, y$ for a random public key $pk$ and image $y = F(pk, x)$ cannot find a preimage $x'$ such that $F(pk, x') = y$ with better than negligible success probability.
%We establish this by showing that an adversary cannot distinguish this security game from a challenge $pk^*, y^*$ sampled uniformly and independently of the oracle.
%This suffices because the probability that \textit{any} $x'$ exists that satisfies $F(pk^*, x') = y^*$ is negligible.

\begin{lemma}
\label{prop:pseudorandom-pk}
    $(\Gen, \Eval, \Inv)$ in \Cref{prop:towf-correctness} has pseudorandom public keys with probability $1$ over $A \sim \mathcal{E}_{TD}$.
\end{lemma}
\begin{proof}
    In this proof, we will work directly with $(G, F, I)$ instead of $(\Gen, \Eval, \Inv)$, because the Forrelation-decoding algorithm in \Cref{prop:towf-correctness} correctly evaluates $(G, F, I)$ pseudo-deterministically on all but finitely many inputs.
    %Throughout this section, we will simply work with $(F, G, I)$ as the trapdoor one-way function.
    %\comment{what does this mean? Q: TODO can we just get rid of injectivity so that this part is no longer needed? W: my TODO}, which only differs from the real $(\Gen, \Eval, \Inv)$ on finitely many points.
    Since we only care about whether the adversary's advantage is negligible (a tail event), a finite change to the construction does not impact security.

    Fix any quantum polynomial time oracular adversary $\mathcal A$.
    To prove the lemma, it suffices to bound
    \[
    \E_{A \sim \mathcal{E}_{TD}}\abs{\Pr_{td \sim \{0,1\}^{n}}\mbracket{\mathcal A^{\mathcal O[A]}(G(td)) = 1} - \Pr_{pk^* \sim \{0,1\}^{\lambda}}\mbracket{\mathcal A^{\mathcal O[A]}(pk^*) = 1}}
    \]
    with a negligible quantity for uniformly random $td, pk^*$ by following the same argument as found in the proof of \Cref{thm:main_prf}.
    (We drop $1^n$ as an input for notational simplicity, because public keys are $m > n$ bits long.)
    Using the same trick as used in the proof of \Cref{cor:security_with_abs}, it suffices to instead bound
    \[
    \abs{\Pr_{A, td}\mbracket{\mathcal A^{\mathcal O[A]}(G(td)) = 1} - \Pr_{A, pk^*}\mbracket{\mathcal A^{\mathcal O[A]}(pk^*) = 1}}.
    \]

    We first claim that this is negligible even if we only count the queries that $\mathcal A$ makes to $\overline{G}$ and $\overline{I}$ on length $n$.
    That is, $\mathcal A$ is allowed to know the entire rest of the oracle, and even the unencoded $F$.
    As a special case, we may assume that the adversary has access to the function $F_{pk}^{-1}$ that inverts $F$ for the given key $pk$, since $F_{pk}^{-1}$ is already determined by $F$ and $pk$.\footnote{Looking ahead slightly, $(pk, F_{pk}^{-1})$ will play the role of $z$ in \Cref{lem:pseudorandom_image}.}
    This means that we consider
    \[
    \abs{\Pr_{A, td}\mbracket{\mathcal A^{\mathcal O[A]}(G(td), F_{G(td)}^{-1}) = 1} - \Pr_{A, pk^*}\mbracket{\mathcal A^{\mathcal O[A]}(pk^*, F_{pk^*}^{-1}) = 1}}.
    \]
    This quantity can only be larger since now $\mathcal A$ gets more information about the input.
    
    We next view the sampling of $(G, I)$ equivalently as follows: for each $td$, sample a random $pk$ and set $G(td) = pk$ and $I(td, \cdot) = F_{pk^*}^{-1}$.
    From this view, we claim that
    \[
    \Pr_{A, td}\mbracket{\mathcal A^{\mathcal O[A]}(G(td), F_{G(td)}^{-1}) = 1} = \Pr_{A, pk^*, td, A^*}\mbracket{\mathcal A^{\mathcal O[A^*]}(pk^*, F_{pk^*}^{-1}) = 1}
    \]
    where $A^*$ is the same as $A$ but we set $G(td) = pk^*$ and $I(td, \cdot) = F_{pk^*}^{-1}$, and then we resample from the corresponding patterned Forrelation distribution for $\overline G(td, \cdot)$ and $\overline I(td, \cdot)$.
    %The claim follows by observing that $A, td$ is identically distributed as $A^*, td$.
    In light of this, we can bound:
    %Further observe that we can also instead perform the following equivalent sampling: first $A$ and $pk^*$ are independently sampled, then $A^*$ is sampled conditioned on $G(td) = pk^*$.
    %Therefore, we can consider the two experiments happening in the same probability space, i.e.
    \begin{align*}
        &\equad \abs{\Pr_{A, td}\mbracket{\mathcal A^{\mathcal O[A]}(G(td), F_{G(td)}^{-1}) = 1} - \Pr_{A, pk^*}\mbracket{\mathcal A^{\mathcal O[A]}(pk^*, F_{pk^*}^{-1}) = 1}} \\
        &= \abs{\Pr_{A,pk^*, td, A^*}\mbracket{\mathcal A^{\mathcal O[A^*]}(pk^*, F_{pk^*}^{-1}) = 1} - \Pr_{A, pk^*}\mbracket{\mathcal A^{\mathcal O[A]}(pk^*, F_{pk^*}^{-1}) = 1}} \\
        &= \abs{\E_{A, pk^*}\mbracket{\Pr_{td,A^*}\mbracket{\mathcal A^{\mathcal O[A^*]}(pk^*, F_{pk^*}^{-1}) = 1} - \Pr\mbracket{\mathcal A^{\mathcal O[A]}(pk^*, F_{pk^*}^{-1}) = 1}}} \\
        &\le \E_{A, pk^*}\abs{\Pr_{td,A^*}\mbracket{\mathcal A^{\mathcal O[A^*]}(pk^*, F_{pk^*}^{-1}) = 1} - \Pr\mbracket{\mathcal A^{\mathcal O[A]}(pk^*, F_{pk^*}^{-1}) = 1}},
    \end{align*}
    where the last inequality is Jensen's.
    We conclude that this quantity is negligible by invoking \Cref{lem:pseudorandom_image}, setting $\SDist$ to be the distribution over $(pk^*, F_{pk^*}^{-1})$ for a uniformly random $pk^*$, $N = \lambda + 2^m(\ell + 1) = 3n + 2^{6n}(n + 1) \le 2^{6n + o(n)}$, and $L = 2^{15n}$.
\end{proof}

\begin{lemma}
\label{prop:oneway-fake-pk}
    $F$ is one-way with a uniformly random $pk$ (see \eqref{eq:trapdoor-oneway-fake}), with probability $1$ over $A \sim \mathcal{E}_{TD}$.
\end{lemma}
\newcommand{\Hybrid}{\mathsf{Hyb}}
\begin{proof}
%\comment{W:I changed this proof to look more like the last lemma...}
Fix a quantum polynomial-time oracular adversary $\mathcal A$.
We construct a polynomial-time checker algorithm $\mathcal C^{\mathcal O[A]}(pk, y)$ to output $1$ if and only if $F(pk, \mathcal A^{\mathcal O[A]}(pk, y)) = y$ (again we drop the input $1^n$ as it is unnecessary).
Note that this algorithm only needs to run $\mathcal A$ once and then spend $\poly(n)$ additional time and queries to $\mathcal O[A]$ to evaluate $F$.
Then following the same argument as found in the proof of \Cref{thm:main_prf}, it suffices to bound $\Pr_{A, pk, x}\mbracket{\mathcal C^{\mathcal O[A]}(pk, F(pk, x)) = 1}$ for uniformly random $pk, x$.

Given the truth table of $G$ for $A$, let $pk^*$ be uniformly sampled from $\{0, 1\}^\lambda$ conditioned on the event that for all $pk$, $G(td) \neq pk^*$, in other words, it is uniformly sampled from $\{0, 1\}^\lambda \setminus G(\{0, 1\}^n)$.
Since the fraction of the image of $G$ is at most $2^n/2^\lambda = 2^{-2n}$, the statistical distance between $pk$ and $pk^*$ is $2^{-\Omega(n)}$.
Thus
\begin{equation}
    \label{eq:always-bad-pk}
    \abs{\Pr_{A, pk, x}\mbracket{\mathcal C^{\mathcal O[A]}(pk, F(pk, x)) = 1} - \Pr_{A, pk^*, x}\mbracket{\mathcal C^{\mathcal O[A]}(pk^*, F(pk^*, x)) = 1}} \le 2^{-\Omega(n)},
\end{equation}
meaning that it suffices to consider $pk^*$ instead.

We next view the sampling of $F$ equivalently as follows: for each $pk$ and $x$, sample a random $y$ and set $F(pk,x) = y$.
From this view, we claim that
\[
\Pr_{A, pk^*,x}\mbracket{\mathcal C^{\mathcal O[A]}(pk^*, F(pk^*,x)) = 1} = \Pr_{A, pk^*, y^*, x, A^*}\mbracket{\mathcal C^{\mathcal O[A^*]}(pk^*, y^*) = 1}
\]
where $A^*$ is the same as $A$ but we set $F(pk^*,x) = y^*$, and then we sample from the corresponding patterned Forrelation distribution for $\overline F(pk^*,x, \cdot)$.
%The claim follows by observing that $A, td$ is identically distributed as $A^*, td$.
In light of this, we can bound:
\begin{align}
&\equad \abs{\Pr_{A, pk^*, x}\mbracket{\mathcal C^{\mathcal O[A]}(pk^*, F(pk^*, x)) = 1} - \Pr_{A,pk^*, y^*}\mbracket{\mathcal C^{\mathcal O[A]}(pk^*, y^*) = 1}} \nonumber\\
&= \abs{\Pr_{A, pk^*, y^*, x, A^*}\mbracket{\mathcal C^{\mathcal O[A^*]}(pk^*, y^*) = 1} - \Pr_{A,pk^*, y^*}\mbracket{\mathcal C^{\mathcal O[A]}(pk^*, y^*) = 1}} \nonumber\\
&= \abs{\E_{A, pk^*, y^*}\mbracket{\Pr_{x,A^*}\mbracket{\mathcal C^{\mathcal O[A^*]}(pk^*, y^*) = 1} - \Pr\mbracket{\mathcal C^{\mathcal O[A]}(pk^*, y^*) = 1}}} \nonumber\\
&\le \E_{A, pk^*, y^*}\abs{
    \Pr_{x,A^*}\mbracket{\mathcal C^{\mathcal O[A^*]}(pk^*, y^*) = 1} - \Pr\mbracket{\mathcal C^{\mathcal O[A]}(pk^*, y^*) = 1}
}\nonumber\\
&= \negl(n).
\label{eq:pseudorandom-image-towf}
\end{align}
Above, the penultimate inequality is Jensen's, and the last inequality is via \Cref{lem:pseudorandom_image}, considering $z = y^*$, $N = m = 6n$ and $L = 2^{15n}$.
%\comment{W: I guess this is the only place where $\mathcal{S}$ depends on the rest of $A$, because we only sample outside the range of $G$? Q: $\SDist$ is used actually in the other indistinguishability proof where the inversion $I$ part of the row needs to depend on $F$, here resampling only rows in $F$}

Finally, the probability that a random $y^*$ is in the image of $F(pk^*, \cdot)$ is at most $2^n/2^m = 2^{-5n}$, thus
\begin{equation}
    \label{eq:random-image-uninvertible}
    \Pr_{pk^*,y^*}\mbracket{\mathcal C^{\mathcal O[A]}(pk^*, y^*) = 1} \le 2^{-5n}.
\end{equation}
Gathering \eqref{eq:always-bad-pk}, \eqref{eq:pseudorandom-image-towf}, and \eqref{eq:random-image-uninvertible}, we conclude that the probability that $\Pr_{A, pk, x}\mbracket{\mathcal C^{\mathcal O[A]}(pk, F(pk, x)) = 1}$, or equivalently that $\mathcal A$ inverts, is negligible by triangle inequality.
\end{proof}

\begin{theorem}
    \label{thm:main_towf}
    With probability $1$ over $A \sim \mathcal{E}_{TD}$, there exists a quantum-computable trapdoor one-way function with pseudorandom public keys, relative to $\mathcal{O}[A]$.
\end{theorem}
\begin{proof}
    \Cref{prop:towf-correctness} shows that $(\Gen, \Eval, \Inv)$ satisfies all the properties except one-wayness and pseudorandomness of the public keys.
    The latter is shown in \Cref{prop:pseudorandom-pk} and the former is shown by combining \Cref{prop:pseudorandom-pk,prop:oneway-fake-pk,fact:fake-pk-conversion}.
    % Gathering the above, we conclude that
    % \[\abs{p_0(\mathcal A) - 2^{-2n}} = \abs{\Pr_{(td, pk, x, y, A) \sim \Hybrid_0}\mbracket{F(pk, \mathcal{A}^{\mathcal{O}[A]}(pk, y)) = y} - \frac1{2^{2n}}}\]
    % is negligible.
    % By Markov's inequality, we get that except with negligible probability over $A$, $\Pr_{pk, y}\mbracket{F(pk, \mathcal{A}^{\mathcal{O}[A]}(pk, y)) = y}$ is negligible.
    % Therefore, the Borel--Cantelli Lemma implies that with probability $1$ over $A$, $F$ is one-way.
\end{proof}

\section{Classical-Communication Quantum Cryptomania in Algorithmica}

%\comment{Q: Is it better to call this $\mathsf{QCMA}$-Cryptomania? I think the name kinda makes sense but might be too confusing}
In this section, we show that relative to the same oracle, we in fact also have classical-communication Cryptomania
    %\footnote{
    %Recall that Impagliazzo~\cite{Imp95-average} calls Cryptomania the world where ``public-key cryptography is possible'', which we interpret as public-key encryption is possible.
    %While Impagliazzo did not give a formal definition of public-key cryptography, we interpret it as such since Impagliazzo wrote Brassard~\cite{Bra79-relativized} ``gives a relativized world where public-key cryptography is possible'', and Brassard defines a weak version of public-key encryption as public-key cryptography.
    %Indeed, public-key encryption scheme can be constructed from trapdoor one-way functions via standard techniques \cite[Theorem 9]{Yao82-trapdoor}.}
while classical computers live in Algorithmica.
Recall that Impagliazzo~\cite{Imp95-average} calls Algorithmica the world where ``$\mathsf{P} = \mathsf{NP}$'' and Cryptomania where ``public-key cryptography is possible''.
While Impagliazzo did not give a formal definition of public-key cryptography, a few concrete primitives were mentioned.
Among them, the strongest two are:
\begin{enumerate}
    \item (Classical-communication) ``secret-key exchange using trap-door one-way functions'': ``by means of conversations in class, Grouse and his pet student would be able to jointly choose a problem that they would both know the answer to, but which Gauss could not solve'';
    \item Secure multiparty computations: ``Any group of people can agree to jointly compute an arbitrary function of secret inputs without compromising their secrets.''%\comment{do we ever instantiate this? I only see OT}
\end{enumerate}

For the rest of this section, we show that there is a \emph{post-quantum fully-black-box reduction} from these (classical-communication) primitives to trapdoor one-way functions with pseudorandom public keys.
We recall that a reduction is fully-black-box \cite{RTV04-reduction} if both the construction and the adversary are used as a ``black-box'' subroutine.\footnote{Strictly speaking, we only need a weaker notion of reductions called (post-quantum) relativizing reductions \cite[Definition 2.7]{RTV04-reduction}.}
We say a (fully-black-box) reduction is post-quantum if the reduction considers classical-input classical-output primitives and works for quantum adversaries potentially having quantum auxiliary inputs.
As a consequence, any post-quantum fully-black-box reduction trivially extends if the primitive is quantum-computable instead.

We begin by constructing semantically secure (classical-communication) public-key encryption schemes, which implies secret-key exchange by using it to encrypt the secret key.
\newcommand{\Enc}{\mathsf{Enc}}
\newcommand{\Dec}{\mathsf{Dec}}

\begin{definition}
    A public-key encryption scheme is a tuple of algorithms $(\Gen, \Enc, \Dec)$ satisfying the following conditions:
    \begin{enumerate}[(i)]
        \item (Input-output behavior) $\Gen(1^n)$ outputs a pair $(pk, sk) \in \{0, 1\}^\lambda \times \{0, 1\}^n$; $\Enc(pk, b)$ encrypts a bit $b$, whose output is a ciphertext $c$; and finally, $\Dec(sk, c)$ outputs the decrypted bit.
        \item (Correctness) For any $b = 0, 1$,
        \[\Pr_{(pk, sk) \gets \Gen(1^n)}\mbracket{\Dec(sk, \Enc(pk, b)) = b} = 1 - \negl(n).\]
        \item (Semantic security) For $(pk, sk) \gets \Gen(1^n)$, $(pk, \Enc(pk, 0))$ is computationally indistinguishable from $(pk, \Enc(pk, 1))$.
    \end{enumerate}
\end{definition}

\begin{fact}
    \label{fact:pke-from-towf}
    There exists a post-quantum fully-black-box reduction from semantically secure public-key encryption schemes to trapdoor one-way functions.
\end{fact}
\begin{proof}
    Consider the standard reduction through hardcore predicates: $\Gen$ is identical except relabeling $td$ as $sk$, $\Enc(pk, m) = (\Eval(pk, x), r, x \cdot r \oplus m)$ for random bitstrings $x, r$, and $\Dec(sk, (y, r, c)) = \Inv(sk, y) \cdot r \oplus c$.
    This is known to be post-quantum \cite{AC01-quantumgl}.
\end{proof}

\begin{fact}
    There exists a post-quantum fully-black-box reduction from secret-key exchange to semantically secure public-key encryption schemes.
\end{fact}
\begin{proof}
    The standard reduction simply asks Alice to send her public key and Bob to use the public key to encrypt random bits for $n$ times.
    Completeness is satisfied since Alice can decrypt by correctness of the encryption scheme.
    By semantic security and hybrid argument, the eavesdropper's view is indistinguishable from $(pk, \Enc(pk, 0^n))$, which is independent of the key exchanged.
\end{proof}

We now turn to classical-communication secure multiparty computations.
Since secure multiparty computations for any classical functionality (in the dishonest majority setting) can be built by calling an entire oblivious transfer (OT) protocol as a black-box \cite{Kil88-ot}, it suffices to simply build OT instead.

Our construction of OT works in two steps as is standard.
We first construct semi-honest oblivious transfer and then upgrade it to malicious security.

\begin{definition}
    A semi-honest oblivious transfer (OT) protocol is a two-party protocol involving a sender and a receiver.
    The sender chooses two input bits $x_0, x_1$; the receiver chooses a choice bit $y$ and receives a single-bit output $z$ from running the protocol.
    Let each party's view be a tuple of its input, its randomness, and the protocol transcript.
    We require the following:
    \begin{enumerate}[(i)]
        \item (Completeness) For any $x_0, x_1, y$, $\Pr[z = x_y]$ is negligibly close to 1.
        \item (Pseudodeterminism) For each party, given its input and randomness and the transcript so far, there exists a unique $m$ that is going to be picked as its next message with probability negligibly close to 1.\footnote{Classically, this is trivially satisfied with probability $1$. However, for a quantum-computed classical-communication oblivious transfer protocol, it is possible that a party may sample randomness using its own quantum computer.}
        \item (Security against semi-honest sender) For any $x_0, x_1$, the sender's view when $y = 0$ is computationally indistinguishable from that when $y = 1$.
        \item (Security against semi-honest receiver) For any $y, x_y$, the receiver's view when $x_{1 - y} = 0$ is computationally indistinguishable from that when $x_{1 - y} = 1$.
    \end{enumerate}
\end{definition}

While semi-honest security is subtler in the quantum setting (see e.g.\ \cite{BCQ23-efi}), as long as the protocol is pseudodeterministic as defined above, the same definition as classically suffices for our setting.
Note that pseudodeterminism is important since a post-quantum fully-black-box reduction is only black-box in terms of invoking the next-message functions as a black box.
It is possible that such reductions may not be black box in terms of functionality (like \cite{Kil88-ot}), meaning that it may not necessarily run and use the input-output behavior of the whole protocol as a black box.
For example, to upgrade a classical semi-honest protocol to malicious security, a coin-flipping protocol is usually needed to set each party's random coins.

On the other hand, since the error probability for pseudodeterminism is negligible, this does not affect any guarantee up to a small negligible loss.

We now construct semi-honest oblivious transfer in a standard way.
Specifically, the reduction shown in the work by Gertner, Kannan, Malkin, Reingold, and Viswanathan \cite{GKMRV00-pkeot} to \emph{fakeable} trapdoor one-way functions turns out to be post-quantum and fully black-box.
For completeness, we reproduce it below.

\begin{lemma}[{\cite[Proposition 12]{GKMRV00-pkeot}}]
    There exists a post-quantum fully-black-box reduction from semi-honest oblivious transfer protocols to trapdoor one-way functions with pseudorandom public keys.
\end{lemma}
\begin{proof}
    We first construct (classical-communication) public-key encryption $(\Gen, \Enc, \Dec)$ from trapdoor one-way functions with semantic security as in \Cref{fact:pke-from-towf}.
    Since $\Gen$ algorithms are identical, this public-key encryption scheme has pseudorandom public keys as well.
    The protocol is as follows:
    \begin{enumerate}
        \item %Let the receiver's input be the choice bit $y$.
          The receiver samples $(pk_y, td_y) \gets \Gen(1^n)$ and samples $pk_{1 - y}$ uniformly at random; and sends $pk_0, pk_1$ to the sender.
        \item The sender encrypts $c_b \gets \Enc(pk_b, x_b)$ for each $b = 0, 1$; and sends $c_0, c_1$.
        \item The receiver decrypts $\Dec(td_y, c_y)$.
    \end{enumerate}
    This protocol is pseudodeterministic since it is a classical protocol other than possibly the use of (quantum-computable) pseudodeterministic primitives.
    Its completeness follows from the correctness of the encryption scheme.
    Security against semi-honest sender holds by pseudorandomness of the public key.
    Finally, security against semi-honest receiver holds assuming semantic security under random public keys as $pk_{1 - y}$ is uniformly sampled; on the other hand, semantic security under random public keys holds by combining standard semantic security and pseudorandomness of the public key (see \Cref{fact:fake-pk-conversion}).
\end{proof}

Now we bootstrap this protocol into a maliciously secure one.
Note that the most standard approaches such as the GMW compiler \cite{GMW87-mpc} would not immediately work.
To see this, we re-emphasize the two requirements.
\begin{enumerate}
    \item The compiler needs to be ``black-box'' in the sense that it should be insensitive to the presence of a quantum-computable primitive.
        Compilers like \cite{GMW87-mpc} are not even relativizing since they use zero-knowledge protocols to prove statements about the underlying protocol.
        The underlying protocol here involves computing a quantum-computable cryptographic primitive, thus it appears that a zero-knowledge protocol for $\mathsf{QCMA}$ must be used instead.
        However, it is unknown how to construct a \emph{classical-communication} zero-knowledge protocol for $\mathsf{QCMA}$ from assumptions like trapdoor one-way functions or even oblivious transfer.
    \item The reduction must be post-quantum, i.e.\ it should remain secure even in the presence of a quantum attacker.
        For example, the fully black-box reduction from \cite{HIKLP11-bbmpc} heavily uses rewinding in the security proof.
\end{enumerate}

Therefore, we invoke a different reduction due to Chatterjee, Liang, Pandey, and Yamakawa \cite[Theorem 1]{CLPY24-pqmpc}, who showed that there exists a \emph{post-quantum fully black-box reduction} from oblivious transfer protocols to semi-honest oblivious transfer and one-way functions.
This work builds on an earlier work \cite{CCLY22-pqextraction} which achieves a similar result except the oblivious transfer protocol is only $\varepsilon$-simulation secure instead of negligibly secure.
Thus, we arrive at the following corollary.

\begin{corollary}
    There exists a post-quantum fully-black-box reduction from oblivious transfer protocols to trapdoor one-way functions with pseudorandom public keys.
    %If there exists quantum-computable one-way functions with pseudorandom public keys, then there exists a (maliciously-secure) oblivious transfer protocol with classical communication.
\end{corollary}

\begin{corollary}
    \label{cor:cryptomania}
    There is a classical oracle relative to which classical-communication oblivious transfer protocols and public-key encryption schemes exist yet $\mathsf{P} = \mathsf{NP}$.
\end{corollary}

\fi

\ifauthors
\section*{Acknowledgments}

We thank Ran Canetti, Vipul Goyal, and Xiao Liang for the references on classical and post-quantum oblivious transfer protocols.
\fi

\bibliographystyle{alphaurl}
\bibliography{MainBibliography}

\end{document}